\documentclass[11pt,a4paper]{article}
\usepackage[english]{babel} 
\usepackage{amsmath}
\usepackage{mathtools}
\usepackage{amsfonts}
\usepackage{amssymb}
\usepackage{amsthm}
\usepackage{makeidx}
\usepackage{graphicx}
\usepackage{natbib}
\usepackage{setspace}
\usepackage{float}
\usepackage{hyperref}
\usepackage[margin=1in]{geometry}
\usepackage{setspace}
\usepackage[titletoc]{appendix}
\usepackage{cleveref}
\usepackage{systeme}
\usepackage{relsize}
\usepackage{multirow}
\usepackage{ragged2e}
\usepackage[dvipsnames,usenames]{color}
\allowdisplaybreaks
\usepackage{subfigure}
\usepackage{authblk}

\hypersetup{
	colorlinks   = true,
	citecolor    = blue,
	urlcolor = blue
}

\newtheorem{Def}{Definition}
 \newtheorem{Thm}{Theorem}
 \newtheorem{Prp}{Proposition}
 \newtheorem{Lmm}{Lemma}
 \newtheorem{Crl}{Corollary}
 \theoremstyle{definition}
 \newtheorem{Exm}{Example}
 \newtheorem{Asp}{Assumption}
 \theoremstyle{remark}
 \newtheorem{Rmk}{Remark}

\DeclareMathOperator*{\argmin}{arg\,min}

\DeclareMathOperator*{\esssup}{ess\,sup}

\DeclareMathOperator*{\cl}{cl}

\newcommand{\E}{{E}}
\newcommand{\R}{\mathbb{R}}

\newcommand{\Z}{\mathbb Z}

\title{Robust risk measures: an averaging approach}
\date{}

\author[1]{Marcelo Righi}
\author[2]{Rodrigo S. Targino}
\affil[1]{Federal University of Rio Grande do Sul (UFRGS), Porto Alegre, RS, Brazil}
\affil[2]{School of Applied Mathematics (EMAp), Getulio Vargas Foundation (FGV), Rio de Janeiro, RJ, Brazil}

\begin{document}

\maketitle

\begin{abstract}
We develop an averaging approach to robust risk measurement under payoff uncertainty. Instead of taking a worst-case value over an uncertainty neighborhood, we weight nearby payoffs more heavily under a chosen metric and average the baseline risk measure. We prove continuity in the neighborhood radius and provide a stable large-radius behavior. In Banach lattices, the approach leads to a convex risk measure and  under separability of the space, a dual representation through a penalty term based on an inf-convolution taken over a Gelfand integral constraint. We also relate our averaging to aggregation at the distribution and quantile levels of payoffs, obtaining dominance and coincidence results. Numerical illustrations are conducted to verify calibration and sensitivity.

\end{abstract}

\section{Introduction\label{sec:intro}}
Knightian uncertainty plays a central role in risk management, as decision-makers face the consequences of their risk assessments with partial knowledge of probabilities and random variables. Since the seminal work of \cite{Artzner1999}, risk measures have become a standard tool in mathematical finance, and it is customary to adopt a worst-case approach to address such uncertainty. This perspective is followed, for example, in \cite{Wang2018}, \cite{Bellini2018}, \cite{Fadina2023}, and \cite{miao2025robust}, which focus on uncertainty in the underlying probability measure, as well as in \cite{Liu2022}, \cite{Hu2024}, \cite{Moresco2023}, \cite{Cornilly2019}, \cite{Cai2023}, \cite{Pesenti2020A}, \cite{Chen2021}, and \cite{Zuo2024}, where the emphasis is on uncertainty in the random variables themselves. 

 The set-up involving uncertainty about the random variables is closely linked to the uncertainty and risk of the model. In this case, distinctly to what generally happens for the uncertainty regarding scenarios, the uncertainty is more intricate  since there are uncertainty sets specific to any point in the domain. In a general fashion, letting $\mathcal{X}$ be some space of random variables $X$, when $\rho\colon\mathcal{X}\to\R$ is a risk measure, its worst-case version under a family of non-empty uncertainty sets $U(X)\subseteq 
\mathcal{X}$ is computed as  \[ \rho^{WC}(X)=\sup\limits_{Z\in U(X)}\rho(Z).\]

The worst-case formulation, while theoretically appealing for its robustness, is often overly conservative and can lead to inefficient or excessively risk-averse decisions. By focusing solely on the most adverse scenario, it disregards valuable information about intermediate or more plausible models, treating all model misspecifications as equally catastrophic. A more balanced approach is to consider integral combinations of risk measures, in which multiple risk evaluations are aggregated through a continuous weighting scheme. In this paper we propose a robust risk measure based on averaging as \[\rho_{\mu,r}(X)=\int_{\bar{B}(X,r)}\rho(Z)d\mu_{X,r}(Z),\] where $\mathcal{X}$ is a metric space equipped with Borel sigma-algebra, with $\mu_{X,r}$ a probability measure supported on the closed ball $\bar{B}(X,r)$, whose density depends on the distance to the center $X$ measured by 
a decreasing function. The point here is to give more weight (importance) to random variables that are closer to the point of the domain.

Beyond the introduction of the averaging as a reasoning for robustness in risk measures, our main theoretical contribution is threefold:
\begin{enumerate}
   \itemsep0em 
    \item A typical point of criticism in the worst-case approach is that the value for the robust risk measure is too sensitive for the choice of the uncertainty parameter, in this case the radius $r$. In particular, the value becomes impracticable when $r\to\infty$. In \Cref{thm:r}, we prove that under Lipschitz continuity, which is valid for most used risk measures in the suitable space, we have a smooth (continuous) behavior for $r$, including the return to the base risk measure when uncertainty is mitigated. Moreover, under mild conditions such as integrability over $\mathcal{X}$,  the averaging approach for robust risk measure exhibits a well-behaved asymptotic pattern. This feature allows, for instance, one to eliminate the parameter $r$ if needed.
    \item A natural question that arises is if we can use $\rho_{\mu,r}$ as a risk measure, instead of only a robust averaging. To answer that, we consider the case where $\mathcal{X}$ is a Banach lattice in order to investigate results on the axiomatic theory of risk measures. In \Cref{Thm:main}, we prove that $\rho_{\mu,r}$ inherits the properties of a convex risk measure from $\rho$ and, under separability of $\mathcal{X}$, we  characterize a penalty function for its dual representation. The penalty is linked to the concept of inf-convolution between the penalty term for $\rho$ and the bilinear form for the dual pair $(\mathcal{X},\mathcal{X}^\prime)$, relying on the concept of Gelfand integral over the dual set. As a consequence, we have the lower bound $\rho\leq\rho_{\mu,r}$ and characterize the associated acceptance set and sub-differential of $\rho_{\mu,r}$.
    \item In most of literature for uncertainty, the focus is more on distributions than random variables per se.  In this context, a valid question is whether to aggregate first is a suitable alternative to averaging, i.e. by considering the functional $X\mapsto \rho(X_{\mu,r})$, where $X_{\mu,r}$ is defined through integration of distributions over closed ball $\bar{B}(X,r)$. More specifically, what is the relationship, if any, between this aggregated approach and averaging $\rho_{\mu,r}$? Moreover, similarly to aggregating the distribution, one can think of aggregating quantiles. In \Cref{thm:avg}, we show that under risk aversion on distributions we have that $\rho_{\mu,r}$ is dominated above by the aggregation in distributions approach, and that for spectral risk measures, the aggregated quantiles approach coincides with $\rho_{\mu,r}$. As a consequence, when $\mathcal{X}$ is contained on the space of integrable random variables stochastic dominance relations are derived.
\end{enumerate}

In addition to the theoretical results, we provide numerical 
illustrations in \Cref{sec:gaussian_examples} that complement 
the theory from three angles. First, we study the averaging 
measure under Gaussian base measures in a Hilbert space with 
linear and quadratic risk functionals, deriving explicit 
finite-dimensional approximations and showing that the 
infinite-dimensional limit exhibits a curse-of-dimensionality 
phenomenon: any ball of fixed radius captures a vanishing 
fraction of the Gaussian measure, so the conditioning 
concentrates entirely at the center regardless of $r$. 
Second, we illustrate the Bayesian interpretation of the 
construction, using a Normal-Gamma base measure over a 
Gaussian family endowed with the Wasserstein metric, and 
verify empirically a dominance ordering between the averaging 
measure, its aggregated counterparts, and the worst-case 
--- finding support for the ordering even in a setting not 
fully covered by our theory. Third, within the same Bayesian 
example we show that the averaging measure grows smoothly and 
sub-linearly in the uncertainty radius, in sharp contrast to 
the worst-case which grows linearly, and study sensitivity to 
the kernel and prior parameters, showing that they provide 
complementary robustness dials controlling the reach and shape 
of the uncertainty weighting.

As for the contribution over the current literature, to the best of our knowledge there is no study proposing such an approach and proving the results we have here. Studies that consider some class of averaging in risk measures such  \cite{Rockafellar2014}, \cite{Rockafellar2018}, \cite{Drapeau2019}, \cite{tadese2020relative},  \cite{Wang2018}, \cite{Righi2019}, \cite{Righi2023}, \cite{Mu2024}, and \cite{Liu2026}, for instance, do so by fixing the point $X$ and placing uncertainty over the risk measure, while we fix the risk measure and place uncertainty over $X$, a more complicated task. From a decision theory perspective, the construction we propose mirrors the representation of ambiguity-averse preferences and model misspecification in \cite{klibanoff2005smooth}, \cite{gilboa2010objective} and \cite{cerreia2025making}, for instance. Nonetheless, such works are not focused in risk measures or a parametric approach as in our case. From a set analysis framework, our $\rho_{\mu,r}$ can be understood as a Set Risk Measure as proposed in \cite{Righi2025}, i.e. a real valued map over the closed and bounded sets of $\mathcal{X}$. However, the focus in here is taken from the usual scalar monetary risk measure, with features that are not taken into consideration in their more abstract setup. The aggregation of distributions is conducted for the worst-case approach by \cite{Mao2025} and \cite{vanoosten2025}, for instance. Thus,  the mixing in an integral form as in our case is distinct from the framework of the aforementioned authors.

The remainder of this paper is structured as follows: In \Cref{sec:prel} we introduce the mathematical framework of our approach and the asymptotic patterns regarding the radius in \Cref{thm:r}. In \Cref{sec:main} we expose our results regarding the convex dual representation of the average robust risk measure in \Cref{Thm:main}. In \Cref{sec:agg} we  establish relationships between $\rho_{\mu,r}$ and aggregated formulations proving bounds that connect them  in the space of distributions. In \Cref{sec:gaussian_examples} the theoretical 
results are complemented by numerical experiments: we study 
the averaging measure under Gaussian base measures in a Hilbert 
space, illustrate the curse of dimensionality in the 
infinite-dimensional limit, and conduct a Bayesian example 
under the Wasserstein metric that verifies the dominance 
ordering, demonstrates the sub-linear growth of the averaging 
measure relative to the worst-case, and examine sensitivity 
to the kernel and prior parameters.

\section{Asymptotics in \texorpdfstring{$r$}{r}}\label{sec:prel}

Let $(\mathcal{X},d)$ be a metric space on some set of real valued random variables $X$ in a measurable space $(\Omega,\mathcal{F})$. Equip it with its usual Borel sigma-algebra $\mathcal{B}=\mathcal{B}(\mathcal{X},d)$.  We define $X^+=\max(X,0)$, $X^-=\max(-X,0)$, and $1_A$ as the indicator function for an event $A\in\mathcal{B}$. For each $X\in\mathcal{X}$  we consider a probability measure $\gamma_X$ on $(\mathcal{X},\mathcal{B})$. For each $r>0$ we then define the measure $\mu_{X,r}$ supported on the closed ball $\bar{B}(X,r)$, absolutely continuous with respect to $\gamma_X$, through the Radon-Nikodym derivative \[\frac{d\mu_{X,r}}{d\gamma_X}(Z)= \frac{1}{K(X,r)} \cdot 1_{\bar{B}(X,r)}(Z) \cdot \varphi(d(X, Z)),\]
where \( \varphi: \mathbb{R}_+ \to \mathbb{R}_+/\{0\} \) is a continuous, decreasing function (e.g., \( \varphi(t) = e^{-\lambda t^2} \)), such that $Z\mapsto\varphi(d(X,Z))$ is bounded above for any $X$. Notice that in this case $Z\mapsto\varphi(d(X,Z))$  is integrable with respect to $\gamma_X$, with integral $K_X\in(0,\infty)$, and we can then define the  normalization  \[K(X,r)=\int_{\bar{B}(X,r)} \varphi(d(X, Z)) \, d\gamma_X(Z).\]  Moreover, with some abuse of notation, we transit between $\mathcal{X}$ and $\bar{B}(X,r)$ with the inherited topology and measurable subspace. Notice that $\mu_{X,r}$ is a truncated Gibbs/tilted version of $\gamma_X$ with potential
$Z\mapsto d(X,Z)$. Of course, the simplest choice is a fixed $\gamma$ for any $X\in\mathcal{X}$.

\begin{Def} \label{def:avg_robust_rm}
Consider a  Borel measurable risk measure $\rho\colon\mathcal{X}\to\mathbb{R}$. Its average robust version is defined as \[\rho_{\mu,r}(X)=\int_{\bar{B}(X,r)}\rho(Z)d\mu_{X,r}(Z).\]
\end{Def}

\begin{Rmk}
 Since most examples of risk measures are at least lower semicontinuous (lsc) in their respective spaces, Borel measurability is not restrictive. Moreover, we assume the average risk measure is real valued, i.e. $\rho$ is integrable in $\bar{B}(X,r)$ with respect to $\mu_{X,r}$ for any $X\in\mathcal{X}$ and any $r>0$. This is the case, for instance, when the function $\varphi$ has a fast decay, or when $\rho$ is bounded in closed balls. Lipschitz continuity (of order $C$) is a sufficient condition for integrability since for any $Z\in\bar{B}(X,r)$ we have $|\rho(Z)|\leq |\rho(X)|+Cr$. As the measure $\mu_{X,r}$ is already defined and supported in $\bar{B}(X,r)$, in order to ease the notation, when there is no ambiguity in interpretation we write \[\rho_{\mu,r}(X)=\int_{}\rho(Z)d\mu_{X,r}(Z).\]
\end{Rmk}

\begin{Rmk} \label{rmk:var}
A more general theoretical approach would be to consider the variational composition introduced in \cite{cerreia2025making} for decision theory as \[\sup\limits_{\mu\in\Phi}\left\lbrace\int_{}\rho(Z)d\mu(Z)-\tau(\mu)\right\rbrace,\] where $\Phi$ is the set of probability measures in $\bar{B}(X,r)$ with the usual Borel sigma-algebra, and $\tau$ is a convex penalty term. Nonetheless, in the current framework of our paper we are interested into a  concrete way of working with such an averaging. We left the variational approach for averaging risk measures for future study.
\end{Rmk}

\begin{Exm}[Discrete perturbation measures]
A convenient class of examples is obtained when, for each \(X\in\mathcal X\), the reference measure \(\gamma_X\) is discrete. In the notation of Definition~\ref{def:avg_robust_rm}, let
\[
\gamma_X=\sum_{i=1}^m a_i(X)\,\delta_{Z_i(X)},
\qquad a_i(X)\ge 0,
\qquad \sum_{i=1}^m a_i(X)=1,
\]
where \(N\in\mathbb N\) (the countable case is analogous), $\delta$ are Dirac measures. Then the normalizing constant \(K(X,r)\) becomes
\[
K(X,r)=\sum_{i=1}^m a_i(X)\,1_{\bar B(X,r)}(Z_i(X))\,\phi\bigl(d(X,Z_i(X))\bigr),
\]
and therefore \(\mu_{X,r}\) is again discrete, namely
\[
\mu_{X,r}
=
\sum_{i=1}^m
\frac{a_j(X)\,1_{\bar B(X,r)}(Z_i(X))\,\phi(d(X,Z_i(X)))}{K(X,r)}
\,\delta_{Z_i(X)}.
\]
Hence \(\rho_{\mu,r}(X)\) is just a weighted average over the atoms of \(\gamma_X\) lying in \(\bar B(X,r)\), that is,
\[
\rho_{\mu,r}(X)
=
\sum_{i=1}^m
\frac{a_j(X)\,1_{\bar B(X,r)}(Z_i(X))\,\phi(d(X,Z_i(X)))}{K(X,r)}
\,\rho\bigl(Z_i(X)\bigr).
\] Under \Cref{asp:lin} and \Cref{lmm:r} below, this construction simplifies considerably. Indeed, it is enough to prescribe a symmetric discrete probability \(\gamma_0\), and then define \(\gamma_X\) by translation. A natural choice is
\[
\gamma_0
=
p_0\delta_0+\sum_{i=1}^{m} p_i\bigl(\delta_{Z_i}+\delta_{-Z_i}\bigr),
\qquad
p_0\ge 0,\quad p_i\ge 0,\quad
p_0+2\sum_{i=1}^{m}p_i=1,
\]
with \(Z_i\neq 0\) for \(i=1,\dots,m\). The atom at \(0\) is written separately since \(0=-0\). Then \(\gamma_0\) is symmetric as needed and the translated measures are given by
\[
\gamma_X
=
p_0\delta_X+\sum_{i=1}^{m} p_i\bigl(\delta_{X+Z_i}+\delta_{X-Z_i}\bigr).
\] In this case, \(\mu_{0,r}\) is obtained by keeping only the atoms with \(\|Z_i\|\le r\), weighting them by \(\phi(\|Z_i\|)\), and normalizing. More precisely,
\[
\mu_{0,r}
=
\frac{
p_0\phi(0)\delta_0
+\sum_{i=1}^{m} p_i\,1_{\{\|Z_i\|\le r\}}\,\phi(\|Z_i\|)\bigl(\delta_{Z_i}+\delta_{-Z_i}\bigr)
}{
p_0\phi(0)+2\sum_{i=1}^{m} p_i\,1_{\{\|Z_i\|\le r\}}\,\phi(\|Z_i\|)
}.
\]Thus, one gets that
\[
\rho_{\mu,r}(X)
=
\frac{
p_0\phi(0)\rho(X)
+\sum_{i=1}^{m} p_i\,1_{\{\|Z_i\|\le r\}}\,\phi(\|Z_i\|)
\bigl(\rho(X+Z_i)+\rho(X-Z_i)\bigr)
}{
p_0\phi(0)+2\sum_{i=1}^{m} p_i\,1_{\{\|Z_i\|\le r\}}\,\phi(\|Z_i\|)
}.
\]
\end{Exm}

 A possible question regarding the setup is the choice of the base probability measure $\gamma_X$ and the function $\varphi$. We now show that the proposed functional is continuous for both choices. We note that $\varphi\in C(\mathbb{R}_+)$, equipped with the supremum norm $\|\cdot\|_\infty$ and $\gamma_X \in ca(\mathcal{X})$, the space of countably additive measures in $\mathcal{X}$, equipped with the total variation norm $\|\cdot\|_{TV}$.

 \begin{Lmm}\label{prp:par}
 Let $A\in\mathcal{B}$.  Then the map \[(\varphi,\gamma)\mapsto\dfrac{\int_{A} \rho(Z)\varphi(d(X, Z))d\gamma(Z)}{\int_{A} \varphi(d(X, Z))d\gamma(Z)}\] is continuous in $(C(\mathbb{R}_+),\|\cdot\|_\infty)\times(ca(\bar{B}(X,r)),\|\cdot\|_{TV})$ at  $(\varphi,\gamma)$ such that $\int_{A} \varphi(d(X, Z))d\gamma(Z) > 0$ and $\rho$ is integrable in $A$ with respect to $\gamma$.  \end{Lmm}

  \begin{proof}
    For continuity in the density, let $\varphi_n\to\varphi$ in the uniform norm on $C(\mathbb{R}_+)$. Then, by the Dominated Convergence Theorem we have \[\int_{A} \varphi_n\left(d(X, Z)\right)d\gamma(Z)\to\int_{A} \varphi\left(d(X, Z)\right)d\gamma(Z).\] Thus, we get that \[\dfrac{1_{A}\cdot \varphi_n}{\int_{A} \varphi_n\left(d(X, Z)\right)d\gamma(Z)}\overset{\|\cdot\|_\infty}\to \dfrac{1_{A}\cdot \varphi}{\int_{A} \varphi\left(d(X, Z)\right)d\gamma(Z)}.\]  Since $\|\varphi_n-\varphi\|_\infty\to0$ and $\int_A \varphi(d(X,Z))\,d\gamma(Z)>0$, we directly have 
$\int_A \varphi_n(d(X,Z))\,d\gamma(Z)\to \int_A \varphi(d(X,Z))\,d\gamma(Z)$. Hence, the
denominators are bounded away from $0$ for $n$ large enough. Therefore the normalized densities
converge in $L^1(\gamma)$, which implies $\|\mu_n-\mu\|_{TV}\to0$. We then have that $\mu^n\overset{\|\cdot\|_{TV}}\to \mu$. Hence, we have the desired convergence.

    For continuity on the base probability measure, let $\gamma_n\to\gamma$ in the total variation norm on $ca(\mathcal{X})$. By continuity of the integral as bi-linear form, we have that\[\int_{A} \varphi\left(d(X, Z)\right)d\gamma_n(Z)\to\int_{A} \varphi\left(d(X, Z)\right)d\gamma(Z).\] Thus, we obtain \[\dfrac{1_{A}\cdot \varphi}{\int_{A} \varphi\left(d(X, Z)\right)d\gamma_n(Z)}\overset{\|\cdot\|_\infty}\to \dfrac{1_{A}\cdot \varphi}{\int_{A} \varphi\left(d(X, Z)\right)d\gamma(Z)}.\] Let $\mu_n$ be the measure generated by each $\gamma_n$. By a similar argument as before for the convergence in the density, we then have that $\mu^n\overset{\|\cdot\|_{TV}}\to \mu$. Hence, we get the desired convergence.
  \end{proof}

 The next result assures the averaging approach exhibits a well-behaved asymptotic pattern regarding the radius $r$.

\begin{Thm}\label{thm:r}
If $\rho$ is Lipschitz continuous, then $r\mapsto\rho_{\mu,r}(X)$ is continuous and  $\lim\limits_{r\to 0}\rho_{\mu,r}(X)=\rho(X)$. Moreover, if $\rho$ is integrable with respect to $\gamma_X$ and there is $\varepsilon>0$ such that
$\gamma_X\!\big(\bar{B}(X,\varepsilon)\big) \;>\;0$ and $\varphi(\varepsilon) \;>\;0$, then  \[\bar{\rho}(X):=\lim\limits_{r\to \infty}\rho_{\mu,r}(X)=\int_{\mathcal{X}}\rho(Z)d\bar{\gamma}_X(Z),\:\text{where}\:\frac{d\bar{\gamma}_X}{d\gamma_X}=K_X^{-1}\varphi(d(X,Z)).\]    
\end{Thm}

\begin{proof}
Fix $X$ and let $L$ be the Lipschitz constant for $\rho$. Let $r_n\to r$. Then $\{r_n\}$ is bounded by $J>0$. Thus, for any $Z\in\bar{B}(X,J)$ and any $n\in\mathbb{N}$ we get, using the Lispchitz continuity of $\rho$, that \begin{align*}|\rho(Z)1_{\bar{B}(X,r_n)}(Z)\varphi(d(X,Z))|\leq (|\rho(X)|+Ld(X,Z))\varphi(0)
\leq (|\rho(X)|+LJ)\varphi(0).\end{align*}  By the Dominated Convergence Theorem, both $K(X,r_n)\mapsto K(X,r)$ and \[\lim\limits_{n\to\infty}\int_{\bar{B}(X,r_n)}\rho(Z)d\mu_{X,r_n}(Z)=\int_{\bar{B}(X,r)}\rho(Z)d\mu_{X,r}(Z).\] Moreover, we have that
\[
\big|\rho_{\mu,r}(X)-\rho(X)\big|
=\left|\int_{\bar{B}(X,r)}(\rho(Z)-\rho(X))\,d\mu_{X,r}(Z)\right|
\leq \sup_{Z\in\bar{B}(X,r)}|\rho(Z)-\rho(X)|.
\]
 For $Z \in \bar{B}(X,r)$ it holds that $d(X,Z)\le r$, and by Lipschitz continuity of $\rho$,
\[
\lim\limits_{r\to 0}\sup_{Z\in\bar{B}(X,r)}|\rho(Z)-\rho(X)| \leq L\lim\limits_{r\to 0}r=0. 
\] Hence, $\lim_{r\to 0}\rho_{\mu,r}(X)=\rho(X)$.

For the second claim, recall that  $Z\mapsto\varphi(d(X,Z))$ is integrable with respect to $\gamma_X$ and $\bar{B}(X,r)\uparrow \mathcal{X}$ as $r\to\infty$. 
Thus, by the Dominated Convergence Theorem we have that
\[
\lim\limits_{r\to \infty}K(X,r)=\lim\limits_{r\to \infty}\int_{\bar{B}(X,r)}\varphi(d(X,Z))d\gamma_X(Z)
=\int_{\mathcal{X}}\varphi(d(X,Z))d\gamma_X(Z)
= K_X<\infty.\] Therefore,
\[
\frac{d\mu_{X,r}}{d\gamma_X}(Z)=\frac{1_{\bar{B}(X,r)}(Z)\,\varphi\!\big(d(X,Z)\big)}{K(X,r)}
\;\longrightarrow\;\frac{\varphi(d(X,Z))}{K_X}
\quad\text{$\gamma_X$-a.s.}
\]
For every $r\ge \varepsilon$ we have $\bar{B}(X,\varepsilon)\subseteq \bar{B}(X,r)$ and therefore
\[
K(X,r)
\ge\;\int_{\bar{B}(X,\varepsilon)}\varphi\!\big(d(X,Z)\big)\,d\gamma_X(Z)
\;\ge\; \varphi(\varepsilon)\,\gamma_X\!\big(\bar{B}(X,\varepsilon)\big)
\;=:\;c_\varepsilon \;>\;0.
\]
In particular, for all $r\ge \varepsilon$,
\[
0\le \frac{1_{\bar{B}(X,r)}(Z)\,\varphi(d(X,Z))}{K(X,r)}
\;\le\; \frac{\varphi(0)}{c_\varepsilon},
\qquad \gamma_X\text{-a.s.}
\]
Hence the pointwise convergence of the Radon-Nikodym derivatives as $r\to\infty$ is dominated. Thus, \Cref{prp:par} applies to yield $\rho_{\mu,r}(X)\to K_X^{-1}\int\rho(Z)\varphi(d(X,Z))\,d\gamma_X(Z)$. This concludes the proof.
\end{proof}

\begin{Rmk}
One can ask if the averaging robust functional  is monotone increasing in $r$ as in the case of the worst-case approach. The following example show that this is not the case.
Let \(\mathcal X=\mathbb R^2\) with Euclidean metric, take the base measure
\(\gamma_X=\mathcal N(0,I_2)\) (standard centered Gaussian), choose the weight
\(\varphi\equiv 1\) (uniform on the ball), and define the risk functional
\(\rho(z_1,z_2)=-z_1\). Fix \(X=(-1,0)\). For \(r>0\), the normalized local measure  \(\mu_{X,r}\) has density with respect to $\gamma_X$ as \(1_{\bar{B}((-1,0),r)}\). Let $E_{\gamma_X}$ be expectation with respect to $\gamma_X$. Thus, the robust average risk is the conditional mean of the first coordinate:
\[
\rho_{\mu,r}(X)
= -\int_{\bar{B}((1,0),r)} z_1\,\frac{1}{K(X,r)}\,\phi_2(z)\,dz
= - E_{\gamma_X}\!\left[\,Z_1 \,\middle|\, Z\in B\!\big((-1,0),r\big)\right],
\]
where \(\phi_2\) is the density of \(\mathcal N(0,I_2)\) and
\(K(X,r)=\gamma_X\big(\bar{B}((-1,0),r)\big)\). By the Lipschitz continuity of \(\rho\), we have by \Cref{thm:r} that $\lim_{r\downarrow 0}\rho_{\mu,r}(X)=\rho(X)=1$. By the large-radius limit (and \(E_{\gamma_X}[Z]=0\)), we obtain $\lim_{r\to\infty}\rho_{\mu,r}(X)= \int_{\mathbb R^2}\rho(z)\,d\gamma_X(z)=0$. Hence the function \(r\mapsto \rho_{\mu,r}(X)\) moves from \(1\) (at small \(r\)) toward \(0\)
(as \(r\to\infty\)). In particular, it cannot be (globally) increasing in \(r\). The same construction works for any centered \(\gamma_X\) (not necessarily Gaussian) and any
radial, nonnegative \(\varphi\) with \(\varphi(0)>0\). One can also take \(\rho(z)=-\langle w,z\rangle\) with any \(w\in\mathbb R^2_{\!+}\). In this case, fix \(X\) with \(\langle w,-X\rangle> E_{\gamma_X}[\langle w,-Z\rangle]=0\) and reach the same conclusion.
\end{Rmk}

\begin{Exm}[Gaussian measure]\label{ex:gaus}
Let $\mathcal{X}$ be a real, separable Hilbert space with inner product
$\langle\cdot,\cdot\rangle$ and norm $\|X\|:=\sqrt{\langle X,X\rangle}$.
Fix an orthonormal basis $\{e_k\}_{k\ge1}\subseteq\mathcal{X}$.
Let $\gamma=\mathcal{N}(0,C)$ be a centered Gaussian measure on $\mathcal{X}$
with covariance operator $C$ satisfying $Ce_k=\lambda_k e_k$, where $\lambda_k>0$
and $\sum_{k=1}^\infty \lambda_k<\infty$.
Then a measurable $T\colon H\to H$ such that $T\sim\gamma$ admits the Karhunen--Lo\`eve expansion
$T=\sum_{k=1}^\infty \sqrt{\lambda_k}\,\xi_k e_k$, where
$\xi_k\stackrel{iid}{\sim}\mathcal{N}(0,1)$, with convergence in the metric and
$\gamma$-a.s.\ in $\mathcal{X}$. From now in this example we write $\gamma_X:=\gamma$ in order to ease notation since we are dealing here with a fixed $X$ and we are not interested, for now, on other properties of $\rho_{\mu,r}$  

For $n\in\mathbb{N}$, let $\mathcal{X}_n:=\mathrm{span}\{e_1,\dots,e_n\}$ and let
$P_n:\mathcal{X}\to\mathcal{X}_n$ be the orthogonal projection, i.e.
$P_n X=\sum_{k=1}^n \langle X,e_k\rangle e_k$.
Introduce the coordinate map $\pi_n:\mathcal{X}\to\mathbb{R}^n$ and the canonical
embedding $i_n:\mathbb{R}^n\to\mathcal{X}_n\subseteq\mathcal{X}$ by
$\pi_n(X):=(\langle X,e_1\rangle,\dots,\langle X,e_n\rangle)$ and
$i_n(u):=\sum_{k=1}^n u_k e_k$. Then $P_n=i_n\circ\pi_n$.
Moreover, the push-forward $\gamma_n:=\pi_{n\#}\gamma$ is Gaussian on $\mathbb{R}^n$
with diagonal covariance $\gamma_n=\mathcal{N}(0,\Sigma_n)$, where
$\Sigma_n=\mathrm{diag}(\lambda_1,\dots,\lambda_n)$.
For $X\in\mathcal{X}$ set $X^{(n)}:=\pi_n(X)\in\mathbb{R}^n$ and define the closed
Euclidean ball centered at $X^{(n)}$ by
$\bar{B}^{(n)}(X,r):=\{u\in\mathbb{R}^n:\ \|u-X^{(n)}\|_{\mathbb{R}^n}\le r\}$.

Fix $X\in\mathcal{X}$ and $r>0$. Define the decreasing sequence of cylinder sets
\[A_n(X,r):=\{Z\in\mathcal{X}:\ \|P_n(Z-X)\|\le r\}.\]
Since $\|P_n V\|\uparrow \|V\|$ for every $V\in\mathcal{X}$, we have
$A_n(X,r)\downarrow \bar{B}(X,r)$.
By continuity from above of the probability measure $\gamma$,
$\gamma(\bar{B}(X,r))=\lim_{n\to\infty}\gamma(A_n(X,r))$.
Furthermore, $A_n(X,r)=(\pi_n)^{-1}(\bar{B}^{(n)}(X,r))$, hence by definition of
push-forward measure,
$\gamma(A_n(X,r))=\gamma_n(\bar{B}^{(n)}(X,r))$, and therefore
$\gamma(\bar{B}(X,r))=\lim_{n\to\infty}\gamma_n(\bar{B}^{(n)}(X,r))$.

Define the projected weights
\[w_n(Z):=\varphi(\|P_n(Z-X)\|)\,1_{\{\|P_n(Z-X)\|\le r\}},
\:
w(Z):=\varphi(\|Z-X\|)\,1_{\{\|Z-X\|\le r\}}.\] Then $\|P_n(Z-X)\|\to\|Z-X\|$ for each $Z$, and
$1_{\{\|P_n(Z-X)\|\le r\}}\downarrow 1_{\{\|Z-X\|\le r\}}$.
By continuity of $\varphi$, $w_n(Z)\to w(Z)$ pointwise. Since
$0\le w_n\le \varphi(0)$, the Dominated Convergence Theorem yields
\[
K^{(n)}(X,r):=\int_{\mathcal{X}} w_n(Z)\,\gamma(dZ)\ \longrightarrow\
\int_{\mathcal{X}} w(Z)\,\gamma(dZ)=K(X,r).
\]

Let $f:\mathcal{X}\to\mathbb{R}$ be integrable with respect to $\gamma$, and define
$f_n:\mathbb{R}^n\to\mathbb{R}$ by $f_n(u):=f(i_n(u))$.
Dominated Convergence Theorem gives
$\int_{\mathcal{X}} f(Z)\,w_n(Z)\,\gamma(dZ)\to \int_{\mathcal{X}} f(Z)\,w(Z)\,\gamma(dZ)$.
Since $K^{(n)}(X,r)\to K(X,r)$ and $K(X,r)>0$ (whenever $\varphi$ is not identically
zero on $[0,r]$), we obtain
\[
\int_{\mathcal{X}} f(Z)\,\mu_{X,r}(dZ)
=\frac{\int_{\mathcal{X}} f(Z)\,w(Z)\,\gamma(dZ)}{\int_{\mathcal{X}} w(Z)\,\gamma(dZ)}
=\lim_{n\to\infty}\frac{\int_{\mathcal{X}} f(Z)\,w_n(Z)\,\gamma(dZ)}
{\int_{\mathcal{X}} w_n(Z)\,\gamma(dZ)}.
\]

Finally, because both $w_n$ and $f\circ i_n$ depend only on $\pi_n(Z)$, the identity
$\gamma_n=\pi_{n\#}\gamma$ implies the finite-dimensional representation
\[
\frac{\int_{\mathcal{X}} f(Z)\,w_n(Z)\,\gamma(dZ)}{\int_{\mathcal{X}} w_n(Z)\,\gamma(dZ)}
=
\frac{\int_{\mathbb{R}^n} f_n(u)\,\varphi(\|u-X^{(n)}\|)\,
1_{\{\|u-X^{(n)}\|\le r\}}\,\gamma_n(du)}
{\int_{\mathbb{R}^n} \varphi(\|u-X^{(n)}\|)\,
1_{\{\|u-X^{(n)}\|\le r\}}\,\gamma_n(du)}.
\]
Equivalently,
\begin{equation*}
\int_{\bar{B}(X,r)} f(Z)\,d\mu_{X,r}(Z)
=
\lim_{n\to\infty}
\int_{\bar{B}^{(n)}(X,r)}
f_n(u)\,\frac{1}{K^{(n)}(X,r)}\,\varphi(\|u-X^{(n)}\|)\,d\gamma_n(u).
\end{equation*}
Hence, this choice of base probability measure $\gamma$ allows to numerically approximate $\rho_{\mu,r}$ by finite-dimensional estimates.
\end{Exm}

\section{\texorpdfstring{$\rho_{\mu,r}$}{rhoMuR} as a risk measure}\label{sec:main}

  When $\mathcal{X}$ is a normed space, we can have some simplification for $\rho_{\mu,r}$ if the following assumption holds. Recall that $\mathcal{B}$ is the Borel sigma-algebra on $\mathcal{X}$.

  \begin{Asp}\label{asp:lin}
  Let $\mathcal{X}$ be a linear space. We assume that $\gamma_0$ is symmetric, i.e. $\gamma_0(A)=\gamma_0(-A)$ for any $A\in\mathcal{B}$, and $\gamma_X(A):=\gamma_0(A-X)$ for any $X\in\mathcal{X}$ and any $A\in\mathcal{B}$.
  \end{Asp}

  \begin{Rmk}
 Assuming $\gamma_0$ symmetric means the perturbation has no directional bias, preventing the averaging
procedure from introducing a systematic drift.
Interpret $\gamma_0$ as the distribution of model errors centered at zero.
To model uncertainty around a position $X$, we use the same error distribution
but re-centered at $X$. Defining $\gamma_X$ as the shift of $\gamma_0$ makes the
local uncertainty cloud look the same around every center, which justifies
reducing computations to the centered case and preserves the expected behavior
under cash shifts. It is easy to verify that the same properties also hold for $\mu_{X,r}$ for any $X\in\mathcal{X}$ and any $r>0$. Notice that the Gaussian measure in \Cref{ex:gaus} for $X=0$ generates a probability measure satisfying the properties of \Cref{asp:lin}, while for general centers $X\neq 0$, we can then directly define
$\gamma_X(A):=\gamma_0(A-X)$, for any $A\in\mathcal{B}$. 
  \end{Rmk}

  \begin{Lmm}\label{lmm:r}
Let $(\mathcal{X},\|\cdot\|)$ be a normed space and \Cref{asp:lin} holds, i.e., $\mu_{X,r}(Z) = \mu_{0,r}(Z-X)$. We have  for any $X\in\mathcal{X}$ and any $r>0$ that
 \begin{equation}
 \rho_{\mu,r}(X)=\int_{\bar{B}(0,r)}\rho(X+Z)d\mu_{0,r}(Z).  
 \end{equation}  
 In addition, if $\rho$ is linear, then $\rho_{\mu,r}(X)=\rho(X)$.
  \end{Lmm}

  \begin{proof}
We claim that $K(X,r)=K(0,r)$ for any $X\in\mathcal{X}$. Recall that for any closed ball we have $\bar{B}(X,r)=\{X\}+\bar{B}(0,r)$. Thus, $W\in \bar{B}(X,r)$ if and only if $W=X+Z$ with $\|Z\|\leq r$. From this fact and \Cref{asp:lin} it follows that $K(X,r)=K(0,r)$ for any $X\in\mathcal{X}$ and any $r>0$. We then get \begin{align*}
\int_{\bar{B}(X,r)}\rho(W)\varphi\left(\|X-W\|\right)d\gamma_X(W)&=\int_{\bar{B}(0,r)}\rho(X+Z)\varphi\left(\|Z\|\right)d\gamma_0(Z)\\
&=K(0,r)\int_{\bar{B}(0,r)}\rho(X+Z)d\mu_{0,r}(Z).
\end{align*} Since $K(X,r)=K(0,r)$, we have \[\rho_{\mu,r}(X)=\int_{\bar{B}(0,r)}\rho(X+Z)d\mu_{0,r}(Z).\]

Now, let $\rho$ be linear. We firstly prove the claim for $\bar{B}(0,r)$. Since unit balls are symmetric and \Cref{asp:lin} holds, by $d(0,Z)=d(0,-Z)$, the weight $Z\mapsto \varphi(d(0,Z))$ is even. Therefore, the induced probability measure $\mu_{0,r}$ is symmetric (i.e., invariant under
$Z\mapsto -Z$). Thus, we have that \[\int_{\bar{B}(0,r)}\rho(Z)d\mu_{0,r}(Z)=\frac{1}{K(0,r)}\int_{\bar{B}(0,r)}\rho(-Z)\varphi(\|-Z\|)d\gamma_0(Z)=-\int_{\bar{B}(0,r)}\rho(Z)d\mu_{0,r}(Z).\] Thus, $\int_{\bar{B}(0,r)}Zd\mu_{0,r}\rho(Z)=0$. For an arbitrary closed ball, we have that \[\int_{\bar{B}(X,r)}\rho(Z)d\mu_{X,r}(Z)=\int_{\bar{B}(0,r)}\rho(X+Z)d\mu_{0,r}(Z)=\rho(X)+\int_{\bar{B}(0,r)}\rho(Z)d\mu_{0,r}(Z)=\rho(X).\]
\end{proof}

  We specialize over $\mathcal{X}$ in order to investigate results on the axiomatic theory of risk measures. More specifically, $\mathcal{X}$ is then taken in this section as a (real) Banach lattice, with dual $\mathcal{X}^\prime$, also a Banach lattice.  We treat elements of $\mathcal{X}^\prime$ as $Q$ to make a contrast to the usual $X,Y,Z$ used for $\mathcal{X}$. When no confusion arises from the context, we treat $\|\cdot\|$ and $\leq$ as the norm and partial order of the appropriated space based on their elements without further mention. 
 We define $\langle\cdot,\cdot\rangle$ as the usual bi-linear form for the dual pair. We adopt the financial convention, i.e. $X\geq0$ is a gain and $X<0$ is a loss. 

 A functional $\rho\colon\mathcal{X}\rightarrow\mathbb{R}$ is a monetary risk measure if it possess the following properties:
	\begin{itemize}
    \itemsep0em 
		\item[] Monotonicity: if $X \leq Y$, then $\rho(X) \geq \rho(Y),\:\forall\: X,Y\in \mathcal{X}$.
		\item[] Translation Invariance: $\rho(X+c)=\rho(X)-c,\:\forall\: X\in \mathcal{X},\:\forall\:c \in \mathbb{R}$.
        
        It is a convex risk measure if it possesses in addition
		\item[] Convexity: $\rho(\lambda X+(1-\lambda)Y)\leq \lambda \rho(X)+(1-\lambda)\rho(Y),\:\forall\: X,Y\in \mathcal{X},\:\forall\:\lambda\in[0,1]$.
		
	\end{itemize}   
 
\begin{Rmk}
 The acceptance set of $\rho$ is defined as $\mathcal{A}_\rho=\left\lbrace X\in  
 \mathcal{X}\colon\rho(X)\leq 0 \right\rbrace $. It is a primal object in the theory of risk measures with a capital adequacy intuition since for monetary risk measures, $\rho(X)=\inf\{m\in\R\colon X+m\in\mathcal{A}_\rho\}$. Let $\mathcal{L}(\mathcal{X})$ be the linear space of measurable real valued maps  on $(\mathcal{X},\mathcal{B})$. In general, we have that \begin{align*}
\mathcal{A}_{\rho_{\mu,r}}&=\bigcup_{\left\lbrace C\in\mathcal{L}(\mathcal{X})\colon\int_{\bar{B}(X,r)\}}C(Z)d\mu_{X,r}\leq 0\right\rbrace}\left\lbrace X\in\mathcal{X}\colon \mu_{X,r}(Z+C(Z)\in\mathcal{A}_\rho)=1\right\rbrace.
\end{align*} If $\rho$ is monetary, then \Cref{lmm:r} implies that \begin{align*}
\mathcal{A}_{\rho_{\mu,r}}&=\bigcup_{\left\lbrace C\in\mathcal{L}(\mathcal{X})\colon\int_{\bar{B}(0,r)\}}C(Z)d\mu_{X,r}\leq 0\right\rbrace}\left\lbrace X\in\mathcal{X}\colon \mu_{X,r}(X+Z+C(Z)\in\mathcal{A}_\rho)=1\right\rbrace.
\end{align*}
\end{Rmk}

We define $\mathcal{X}^\prime_{1,+}:=\{Q\in\mathcal{X}^\prime\colon \langle X,Q\rangle\geq 0\:\forall\:X\geq 0,\:\langle 1,Q\rangle =1\}$. For any $\rho\colon \mathcal{X}\to\mathbb{R}$, its sub-gradient at $X\in \mathcal{X}$ is $\partial \rho(X)=\{Q\in \mathcal{X}^\prime\colon \rho(Z)-\rho(X) \geq \langle Z-X,Q\rangle\:\forall\:Z\in \mathcal{X}\}$.  We say $\rho\colon \mathcal{X}\to\mathbb{R}$ is Gâteaux differentiable at $X\in \mathcal{X}$ when $t\mapsto\rho(X+tZ)$ is differentiable at
$t = 0$ for any $Z\in \mathcal{X}$ and the derivative defines a continuous linear functional on $\mathcal{X}$. In the context of the Gâteaux differential, we treat $Q$, and the continuous linear functional it defines as the same. From the Namioka-Klee Theorem, see Theorem 1 in \cite{Biagini2010} and the ongoing discussion for details,
 $\rho\colon\mathcal{X}\rightarrow \mathbb{R}$ is a convex risk measure on a Fréchet Lattice $\mathcal{X}$, which is always the case for Banach lattices, if and only if it can be represented as
	\begin{equation}\label{eq:dual}
	\rho(X)=\max\limits_{Q\in\mathcal{X}^\prime_{1,+}}\left\lbrace \langle -X,Q\rangle -\alpha_\rho(Q)\right\rbrace,\:\forall\:X\in \mathcal{X},
	\end{equation} where  \begin{equation*}
\alpha_\rho(Q)=\sup\limits_{X\in \mathcal{X}}\{\langle -X,Q\rangle-\rho(X)\}.
	\end{equation*} In this case, $\rho$ is continuous and sub-differentiable. Direct calculation shows that $\partial \rho(X)=\left\lbrace Q\in\mathcal{X}^\prime_{1,+}\colon \rho(X)=\langle -X, Q\rangle -\alpha_\rho(Q)\right\rbrace\neq\emptyset$.

We can change the set where the dual representation supremum is taken without harm.

\begin{Lmm}\label{lemma:bound}
The supremum in \eqref{eq:dual} can w.l.o.g. be taken over the weakly* compact \[\mathcal{X}^*:=\cl^*\left(\{Q\in\mathcal{X}^\prime_{1,+}\colon \alpha_\rho(Q)<\infty\}\right).\]
\end{Lmm}

\begin{proof}
The fact that the supremum in \eqref{eq:dual} is not affected if taken over $\mathcal{X}^*$  is straightforwardly obtained from the fact that we can rule out $\{Q\in\mathcal{X}^\prime_{1,+}\colon\alpha(Q)=\infty\}$ and $Q\mapsto \langle -X,Q\rangle -\alpha_\rho(Q)$ is weakly* upper semicontinuous (usc). Assume towards a contradiction that $\mathcal{X}^*$ is not norm bounded. Then, for all $n\in\mathbb{N}$ there is $(X_n,Q_n)\in\mathcal{X}\times\mathcal{X}*$ such that $X_n\leq 0$, $\|X_n\|=1$ and $\langle -X_n,Q_n\rangle >n^3$. In order to verify such a sequence exists, fix $Q\in X'_+$.
For any $X\in X$ we have $-|X|\le X \le |X|$, thus by positivity of $Q$,
we have $\langle X,Q\rangle \le \langle |X|,Q\rangle$ and $\langle -X,Q\rangle \le \langle |X|,Q\rangle$.
Moreover, in a Banach lattice it does hold that $\||X|\|=\|X\|$.
Therefore,
\[
\sup_{\|X\|\le 1}\langle X,Q\rangle
=\sup_{\substack{\|X\|\le 1\\ X\ge 0}}\langle X,Q\rangle.
\]
Consequently, if $\|Q\|=\sup_{\|X\|\le 1}\langle X,Q\rangle$ is large, we may choose
$Y\ge 0$ with $\|Y\|=1$ and $\langle Y,Q\rangle$ arbitrarily close to $\|Q\|$. Applying this with $Q=Q_n\in X'_{1,+}$, if $\|Q_n\|$ is unbounded we can pick
$Y_n\ge 0$ with $\|Y_n\|=1$ and $\langle Y_n,Q_n\rangle>n^3$.
Setting $X_n:=-Y_n\le 0$ (so $\|X_n\|=1$) yields
$\langle -X_n,Q_n\rangle=\langle Y_n,Q_n\rangle>n^3$.

Let $X:= \sum_{n=1}^\infty \frac{X_n}{n^2}=\lim\limits_{k\to\infty}\sum_{n=1}^k \frac{X_n}{n^2}$. Notice that $X$ is well defined since $\frac{X_n}{n^2}$ is Cauchy and $\mathcal{X}$ is Banach. Also, recall that $\inf_{n\in\mathbb{N}}\alpha_\rho(Q_n)>-\infty$. Then, we get that \[\rho(X)\geq \lim\limits_{n\to \infty}\rho(X_nn^{-2})\geq\lim\limits_{n\to \infty}\left(\langle -X_nn^{-2},Q_n\rangle-\alpha_\rho(Q_n)\right)\geq \lim\limits_{n\to \infty}n-\inf_{n\in\mathbb{N}}\alpha_\rho(Q_n)=\infty. \]Hence, we get a contradiction to finiteness of $\rho$. Thus $\mathcal{X}^*$ is norm bounded, hence weakly* compact.  
\end{proof}

Our next main result shows that $\rho_{\mu,r}$ inherits the properties of a convex risk measure from $\rho$ and we also characterize a penalty function for its dual representation. To that, we need the concept of Gelfand
	integral, see \cite{Aliprantis2006} Chapter 11 for details. Equip $\mathcal{X}^\prime$ with its Borel sigma-algebra from the weak* topology. We have that  $S\colon\mathcal{X}\rightarrow \mathcal{X}^\prime$ is Gelfand (weak*) measurable if for each $X\in\mathcal{X}$ the map $Z\mapsto\langle X,S(Z)\rangle$ is $\mathcal{B}$-measurable. We then have that $Z\mapsto Q_Z$ is Gelfand integrable with respect to a measure $\mu$ over $(\mathcal{X},\mathcal{B})$ if and only if $\langle X, \int_\mathcal{X} Q_Z d\mu \rangle=\int_\mathcal{X}\langle X,  Q_Z  \rangle d\mu$ for any $X\in\mathcal{X}$, where $\int_\mathcal{X} Q_Z d\mu\in\mathcal{X}^\prime$ denotes the Gelfand integral. We can replace $\mathcal{X}$ for any $A\in\mathcal{B}$ as needed.

\begin{Thm}\label{Thm:main}
Let $\rho$ be a convex risk measure and \Cref{asp:lin} holds. Then, $\rho_{\mu,r}$ is a convex risk measure  that inherits differentiability from $\rho$. Moreover, if $\mathcal{X}$ is separable, then it can be represented as \begin{equation}\label{eq:dualmu}\rho_{\mu,r}(X)=\sup_{Q\in\mathcal{X}^\prime_{1,+}}\left\lbrace -\langle X,Q\rangle -\alpha_{\mu,r}(Q)\right\rbrace,\end{equation}
where 	$\alpha_{\mu,r}\colon\mathcal{X}^\prime_{1,+}\rightarrow\mathbb{R}\cup\{\infty\}$ defined as 
		\begin{equation}\label{eq:pennu}
		\alpha_{\mu,r}(Q)=
		\inf\limits_{\substack{\int_{\bar{B}(0,r)}Q_Zd\mu_{0,r}(Z)=Q \\ Z\mapsto\alpha_{\rho}\left(Q_Z\right)-\langle Z, Q_Z\rangle\in\mathcal{B}}}
		\int_{\bar{B}(0,r)}\left(\alpha_{\rho}\left(Q_Z\right)-\langle Z, Q_Z\rangle\right)d\mu_{0,r}(Z).
		\end{equation}
      \end{Thm}

\begin{proof}

By using \Cref{lmm:r}, Monotonicity, Translation Invariance and Convexity of $\rho_{\mu,r}$ follow from those of $\rho$ jointly to
linearity and monotonicity of the integral operator. Thus, from the Namioka-Klee Theorem, $\rho_{\mu,r}$ can be represented by a dual convex conjugate, is continuous and sub-differentiable. For differentiability, fix $X,h\in \mathcal{X}$. For each $Z\in\mathcal{X}$, we have by Gâteaux differentiability of $\rho$ that
\[
g_Z(t)\ :=\ \frac{\rho(X+t h-Z)-\rho(X-Z)}{t}\ \overset{t\to 0}{\longrightarrow}\ D\rho(X-Z)[h].
\]
 By \Cref{lemma:bound}, let $M<\infty$ be the uniform bound of $\mathcal{X}^*$. Thus, we have uniformly on $Z\in\bar{B}(0,r)$ that
$|g_Z(t)|\leq M\,\|h\|$ for small $|t|$. Hence, by the Dominated Convergence Theorem we get
\[
\frac{\rho_{\mu,r}(X+t h)-\rho_{\mu,r}(X)}{t}
=\int_{\bar{B}(0,r)} g_Z(t)\,d\mu_{0,r}(Z)\ \overset{t\to 0}{\longrightarrow} \int_{\bar{B}(0,r)} D\rho(X-Z)[h]\ d\mu_{0,r}(Z),
\]
which is linear and continuous in $h$, i.e., the Gâteaux derivative of $\rho_{\mu,r}$ at $X$.

For the penalty term, for each fixed $Z\in \bar{B}(0,r)$ we have by the dual representation of $\rho$ that
\[
\rho(X+Z)
= \sup_{Q\in\mathcal{X}^*} \Big\{ \langle -(X+Z), Q\rangle - \alpha_\rho(Q) \Big\}
= \sup_{Q\in\mathcal{X}^*} \Big\{ \langle -X,Q\rangle - \alpha_\rho(Q) -\langle Z,Q\rangle \Big\}.
\]
Hence $Z\mapsto \rho(X+Z)=\sup_{Q} h(Z,Q)$ with $h(Z,Q):= \langle -X,Q\rangle - \alpha_\rho(Q) -\langle Z,Q\rangle $,
a point-wise supremum of measurable functions in $Z$ (measurability of $Z\mapsto \langle Z,Q\rangle $ is obvious
and $\alpha_\rho(Q)$ is constant in $Z$). By continuity, $\rho(X+Z)$ is 
$\mathcal{B}$-measurable. Moreover, for each $Z$ we have that  $X\mapsto\rho(X+Z)$ defines a convex risk measure $\rho_Z$ with penalty $\alpha_{\rho_Z}(Q)=\alpha_\rho(Q)-\langle Z,Q\rangle$. By using \Cref{lmm:r}, \Cref{lemma:bound} and the standard inequality  between integral and supremum signs we get that 
\begin{align*}
\rho_{\mu,r}(X)
&= \int_{\bar{B}(0,r)} \sup_{Q\in\mathcal{X}^*} h(Z,Q)\,d\mu_{0,r}(Z) \\
&\ge \sup_{\substack{Z\mapsto h(Z,Q_Z)\in\mathcal{B}}}
\int_{\bar{B}(0,r)} \left(\langle -X, Q_Z\rangle -\alpha_\rho(Q_{\!Z})-\langle Z, Q_Z\rangle\right)\,d\mu_{0,r}(Z).
\end{align*}
Now, we restrict attention to measurable selections $Z\mapsto Q_{\!Z}$ that satisfy the Gelfand integral
constraint $\int_{\bar{B}(0,r)} Q_{\!Z}\,d\mu_{0,r}=Q$. Since in this case $Z\mapsto Q_Z$ is weak$^*$-measurable and essentially bounded with values in
$\mathcal{X}^*$, it is Gelfand integrable and its barycenter
$\int Q_Z\,d\mu(Z)\in\mathcal{X}'$ is well-defined by duality with $\mathcal{X}$.
By using again \Cref{lmm:r}, \Cref{lemma:bound} and taking the supremum over all such measurable selections, we get that
\begin{align*}
\rho_{\mu,r}(X)
&\;\ge\; \sup_{Q\in\mathcal{X}^*}
\left\lbrace \langle -X, Q\rangle
\;-\; \inf_{\substack{Z\mapsto \alpha_\rho(Q_{\!Z})-\langle Z, Q_Z\rangle\in\mathcal{B}\\ \int Q_{\!Z} d\mu_{0,r}=Q}}
\int_{\bar{B}(0,r)} \big(\alpha_\rho(Q_{\!Z})-\langle Z, Q_Z\rangle \big)\,d\mu_{0,r}(Z) \right\rbrace\\
&= \sup_{Q\in\mathcal{Q}}\big\{E_Q[-X]-\alpha_{\mu,r}(Q)\big\}.
\end{align*}

For the converse inequality, fix $\varepsilon>0$. Recall that, by \Cref{lemma:bound}, the supremum in the dual representation of $\rho$ is
attained on the weakly* compact $\mathcal{X}^*$, and $(Z,Q)\mapsto h(Z,Q)$ is Borel measurable in $Z$
and weakly* usc in $Q$. Moreover,  since  $\mathcal{X}^*$ is norm bounded, it is
contained in some ball $k B_{X^*}$, $k>0$, which is weak$^*$-compact. Hence, $\mathcal{X}^*$ is
weak$^*$-compact as a weak$^*$-closed subset of $kB_{X^*}$. As $\mathcal{X}$ is separable, the weak$^*$ topology on
$\mathcal{X}^*$ is metrizable (therefore $(\mathcal{X}^*,\mathcal{B}(\mathcal{X}^*))$ is a
Borel space). Hence, by the Measurable Maximum Theorem, see Theorem 18.19 in \cite{Aliprantis2006} or Proposition A.5 in \cite{bertsekas2012} for instance, for each $\epsilon>0$ there exists a
Gelfand measurable selector $Z\mapsto Q^{\varepsilon}_{\!Z}\in \mathcal{X}^*$ such that
\[
\rho(X+Z)=\sup_{Q\in\mathcal{X}^*} h(Z,Q)\;\le\; h(Z,Q^{\varepsilon}_{\!Z})+\varepsilon
\qquad \text{for $\mu_{0,r}$-a.e.\ }Z.
\]
Since $\mathcal{X}^*$ is norm bounded, we have that the selector is also Gelfand integrable. Let then $Q^{\varepsilon}=\int_{\bar{B}(0,r)} Q^{\varepsilon}_{\!Z}\,d\mu_{0,r}(Z)$. We then get that
\begin{align*}
\rho_{\mu,r}(X)
&\le \int_{\bar{B}(0,r)} \big(\langle -X,Q^{\varepsilon}_{\!Z}\rangle -\alpha_\rho(Q^{\varepsilon}_{\!Z})
      - \langle Z,Q^{\varepsilon}_{\!Z}\rangle\big)\,d\mu_{0,r}(Z) + \varepsilon \\
&= \langle -X, Q^{\varepsilon}\rangle - \int_{\bar{B}(0,r)} \big(\alpha_\rho(Q^{\varepsilon}_{\!Z})-\langle Z,Q^{\varepsilon}_{\!Z}\rangle\big)\,d\mu_{0,r}(Z)
   + \varepsilon.
\end{align*}
By the definition of $\alpha_{\mu,r}$ (infimum over all measurable selections that integrate to the given $Q$),
\[
\int_{\bar{B}(0,r)} \big(\alpha_\rho(Q^{\varepsilon}_{\!Z})-\langle Z,Q^{\varepsilon}_{\!Z}\rangle\big)\,d\mu_{0,r}(Z)
\;\ge\; \alpha_{\mu,r}(Q^{\varepsilon}).
\]
Hence,
\[
\rho_{\mu,r}(X)\;\le\; \langle -X, Q^{\varepsilon}\rangle -\alpha_{\mu,r}( Q^{\varepsilon})+\varepsilon
\;\le\; \sup_{Q\in\mathcal{X}^*}\{\langle -X, Q\rangle-\alpha_{\mu,r}(Q)\}+\varepsilon.
\]
Let then $\varepsilon\downarrow 0$ to obtain
$\rho_{\mu,r}(X)\le \sup_{Q\in\mathcal{X}^*}\{\langle -X, Q\rangle-\alpha_{\mu,r}(Q)\}$. 
This concludes the proof.
\end{proof}

\begin{Rmk}\label{rmk:pen}
The separability assumption is key to the penalty based representation. Nonetheless, most Banach lattices considered in the literature of risk measures can be taken as separable, see the exemples below for instance. Notice we do not need that an optimizer exists for $\eqref{eq:pennu}$, all conclusions only rely on  the existence of measurable selections that are approximately optimal. The penalty $\alpha_{\mu,r}$ is linked to the concept of inf-convolution. It is clear from standard results in convex analysis that $\alpha_{\mu,r}$ is convex and, thus, it is lower semi-continuous if and only if it coincides with  the convex conjugate of $\rho_{\mu,r}$, denoted $\alpha_{\rho_{\mu,r}}$. Thus, we have that $\alpha_{\rho_{\mu,r}}$ is the lsc envelope of $\alpha_{\mu,r}$, i.e. the epigraph of $\alpha_{\rho_{\mu,r}}$ equals the closed epigraph of  $\alpha_{\mu,r}$. 
\end{Rmk}

   \begin{Rmk}\label{rmk:coh}
    A convex risk measure $\rho$ also possesses Positive Homogeneity if $\rho(\lambda X)=\lambda\rho(X)$ for any $\lambda\geq 0$ and any $X\in \mathcal{X}$. In this case we call it coherent. It is easy to verify that $\rho$ is coherent if and only if it can be represented as 			$\rho(X)=\max_{Q\in\mathcal{X}^\prime_\rho} \langle -X,Q\rangle,\:\forall\:X\in \mathcal{X}$, where $Q\in\mathcal{X}^\prime_\rho=\{Q\in\mathcal{X}^\prime_{1,+}\colon \alpha_\rho(Q)=0\}$. In this case $\rho$ is Lipschitz continuous, and hence bounded (in particular integrable) in any closed ball, since for any $X,Y\in\mathcal{X}$ we get $|\rho(X)-\rho(Y)|\leq \rho(|X-Y|)\leq M\|X-Y\|$, where $M$  is the norm bound of $\mathcal{X}_\rho^\prime$. Moreover, we have that $\rho_{\mu,r}$ does not inherits Positive Homogeneity even in the presence of \Cref{asp:lin} because $\lambda\mapsto\rho(\lambda X+Z)$ is not linear. Nonetheless, we have under the conditions of \Cref{Thm:main} the following simplification in the penalty term
\[
\alpha_{\mu,r}(Q)=
\inf\limits_{\substack{\int_{\bar{B}(0,r)}Q_Z\,d\mu_{0,r}(Z)=Q \\ Z\mapsto \langle Z,Q_Z\rangle\in\mathcal{B}\\\mu_{0,r}(Q_Z\in\mathcal{X}^\prime_\rho)=1}}
\int_{\bar{B}(0,r)}\langle -Z,Q_Z\rangle\,d\mu_{0,r}(Z),
\]
if $Q\in\mathcal{X}^\prime_\rho$, and $\alpha_{\mu,r}(Q)=\infty$ otherwise.    
   \end{Rmk} 

\begin{Rmk}
 By combining both \Cref{thm:r} and \Cref{Thm:main} we can have a dual representation with a measure that is the same for any $X\in\mathcal{X}$. More specifically, if $\mathcal{X}$ is separable Banach lattice, $\rho$ is convex risk measure that is integrable with respect to $\gamma_0$ and there is $\varepsilon>0$ such that
$\gamma_0\!\big(\bar{B}(X,\varepsilon)\big) >0$ and $\inf_{t\in[0,\varepsilon]}\varphi(t) >0$, then $\bar{\rho}:=\lim_{r\to\infty}\rho_{\mu,r}$ is a convex risk measure given as $\bar{\rho}(X)=\int_\mathcal{X}\rho(X+Z)d\bar{\gamma}_0(Z),\: \text{where}\:\frac{d\bar{\gamma}_0}{d\gamma_0}=K_0^{-1}\varphi(\|Z\|)$. Moreover, it can be represented by 
		\[
		\bar{\alpha}(Q)=
		\inf\limits_{\substack{\int_{\bar{B}(0,r)}Q_Zd\bar{\gamma}_0(Z)=Q \\ Z\mapsto\alpha_{\rho}\left(Q_Z\right)-\langle Z, Q_Z\rangle\in\mathcal{B}}}
		\int_{\bar{B}(0,r)}\left(\alpha_{\rho}\left(Q_Z\right)-\langle Z, Q_Z\rangle\right)d\bar{\gamma}_0(Z).\]
The penalty term derivation follows the same steps as those in \Cref{Thm:main} by replacing $\mu_{0,r}$ by $\gamma_0$.
\end{Rmk}

A consequence of \Cref{Thm:main} is that we are now able to direct compare $\rho$ and $\rho_{\mu,r}$. In fact we have that $\rho$ a lower bound for $\rho_{\mu,r}$. Moreover, we are now able to better characterize $\partial\rho_{\mu,r}$.

     \begin{Crl}\label{crl:subd}
     Let $\rho$ be a convex risk measure, $\mathcal{X}$ separable, and \Cref{asp:lin} holds. Then,  $\rho_{\mu,r}\geq \rho$ and
\begin{align*}\partial\rho_{\mu,r}(X)\supseteq&
      \left\lbrace Q\in\mathcal{X}^\prime_{1,+}\colon \rho_{\mu,r}(X)= \langle -X, Q\rangle -\alpha_{\mu,r}(Q)\right\rbrace\\
      \supseteq&\left\lbrace Q\in\mathcal{X}^\prime_{1,+}\colon Q=\int_{\bar{B}(0,r)}Q_Zd\mu_{0,r}(Z),\:Q_Z\in\partial\rho(X-Z)\:\mu_{0,r}-a.s.\right\rbrace,\end{align*}with equality if $\rho_{\mu,r}$ is Gateaux differentiable at $X$ and the supremum in \eqref{eq:dualmu} is attained.
     \end{Crl}

    \begin{proof}
    We claim $\alpha_{\mu,r}\leq\alpha_\rho$. If $\alpha_\rho(Q)=\infty$, then the result holds trivially. Let $\alpha_\rho(Q)<\infty$. By taking $Q_Z=Q$ $\mu_{0,r}-a.s.$, we have by \Cref{Thm:main} that \[\alpha_{\mu,r}(Q)\leq \int_{\bar{B}(0,r)}\left(\alpha_{\rho}\left(Q\right)-\langle Z, Q\rangle\right)d\mu_{0,r}(Z)=\alpha_\rho(Q).\]
Hence, $\rho_{\mu,r}\geq\rho$.  

    For the second claim, let $\partial^\prime\rho_{\mu,r}(X)=\{Q\in\mathcal{X}^\prime_{1,+}\colon \rho(X)= \langle -X,Q\rangle -\alpha_{\mu,r}(Q)\}$. By the reasoning in \Cref{rmk:pen}, we always have $\partial^\prime\rho_{\mu,r}(X)\subseteq \partial\rho_{\mu,r}(X)\neq\emptyset$. Now, fix $Q_Z\in\partial \rho(X-Z)$ for any $Z\in\bar{B}(0,r)$ and define $Q:=\int_{\bar{B}(0,r)}Q_Zd\mu_{0,r}(Z)$. By \Cref{Thm:main} we then have that \begin{align*}
    \rho_{\mu,r}(X)&=\int_{\bar{B}(0,r)}\left(\langle -X+Z, Q_Z\rangle -\alpha_\rho(Q_Z)\right)d\mu_{0,r}(Z)\\
    &\leq \langle -X, Q\rangle -\alpha_{\mu,r}(Q)\leq \rho_{\mu,r}(X).
    \end{align*}
    Hence, $Q\in\partial^\prime\rho_{\mu,r}(X)$. For the converse inclusion, if $\rho_{\mu,r}$ is Gateaux differentiable at $X$ and  the supremum in \eqref{eq:dualmu} is attained, then  $\partial\rho_{\mu,r}(X)$ is a singleton and $\partial^\prime\rho_{\mu,r}\neq\emptyset$. Thus, we must to have $\partial^\prime\rho_{\mu,r}(X)=\partial\rho_{\mu,r}(X)$. Furthermore, let $Q\in\partial^\prime\rho_{\mu,r}(X)$ and $\{Q^*_Z\}$ be the inf-convolution minimizer for $Q$. Thus, we have that \begin{align*} \rho_{\mu,r}(X)&=\int_{\bar{B}(0,r)}\left(\langle -X+Z,Q^*_Z\rangle -\alpha_\rho(Q^*_Z)\right)d\mu_{0,r}(Z). \end{align*} Thus, it is necessary that $Q_Z^*\in\partial\rho(X-Z)$ for $\mu_{0,r}-a.s.$ any $Z\in\bar{B}(0,r)$. This concludes the proof.
        \end{proof}

             \begin{Exm}[$\mathcal{X}=L^p$ spaces]\label{ex:lp}
  Fix a probability measure $\mathbb{P}$ in $(\Omega,\mathcal{F})$. Let $L^{p}:=L^{p}(\Omega,\mathcal{F},\mathbb{P})$ be the space of (equivalent classes of) random variables such that  $ \lVert X \rVert_p = (E[|X|^p])^{\frac{1}{p}}<\infty$ for $p\in[1,\infty)$, where $E$ is the expectation over $\mathbb{P}$. Here, the partial order is the $\mathbb{P}-a.s.$ one. As usual, $L^q$, $\frac{1}{p}+\frac{1}{q}=1$ is the dual of $L^p$. In this case, $\mathcal{X}^\prime_{1,+}$ be the set of all probability measures on $(\Omega,\mathcal{F})$ that are absolutely continuous with respect to $\mathbb{P}$, with Radon--Nikodym derivative $\frac{dQ}{d\mathbb{P}}\in L^q$. With some abuse of notation, we treat probability measures as elements of $L^q$. In \cite{Ruszczynski2006} and \cite{Kaina2009}, for instance, convex risk measures are studied in these $L^p$ spaces.  In the work of \cite{Righi2024a}, it is characterized that the worst-case convex risk measure for closed balls in $L^p$ can be represented by \[\alpha_{WC}(Q)=\alpha_\rho(Q)-r\|Q\|_q.\] Since $\rho_{\mu,r}\leq\rho^{WC}$ we have, together to \Cref{crl:subd} the bounds \[\alpha_\rho(Q)-r\|Q\|_q\leq \alpha_{\mu,r}(Q)\leq \alpha_\rho(Q),\:\forall\:Q\in\mathcal{Q}.\]

    Let $L^{\infty}:=L^{\infty}(\Omega,\mathcal{F},\mathbb{P})$ be the space of (equivalent classes of) random variables such that $ \lVert X \rVert_\infty = \esssup |X| < \infty$ for $ p = \infty$. For $L^\infty$, we consider the dual pair $(L^\infty,L^1)$, where the weak topology is in fact its weak* topology. As references for this classical space of risk measures, we refer to \cite{Delbaen2012} and \cite{Follmer2016}. In these spaces order continuity and order duality as in \cite{Biagini2010}, known as Lebesgue continuity, plays a crucial role. A prominent example of convex risk measure in this space is the Entropic (ENT), which is the map $ENT_\gamma\colon L^\infty\to\mathbb{R}$  defined as \[ENT_\gamma (X) = \frac{1}{\gamma} \log \E [ e^{-\gamma X}],\:\gamma>0.\]  Its penalty is the relative entropy as \[\alpha_{ENT_\gamma}(Q)=\frac{1}{\gamma}E\left[\frac{dQ}{d\mathbb{P}}\log\frac{dQ}{d\mathbb{P}}\right].\] This risk measure is Gâteaux differentiable for any $X\in L^\infty$ with $\frac{dQ_X}{d\mathbb{P}}=\frac{e^{-\gamma X}}{E[e^{-\gamma X}]}$. It is easy to verify that such a risk measure is integrable with respect to $\mu_{0,r}$ since closed balls in $L^\infty$ are order bounded.

Concerning the hypothesis of separability, define
$d_\mathcal{F}(A,B):=\mathbb P(A\triangle B)$ on $\mathcal F$. Thus, $(\mathcal{F},d_\mathcal{F})$ is a metric space. Then, $L^p$, $p\in[1,\infty)$, is separable if and only if $(\mathcal{F},d_\mathcal{F})$ is separable, see Theorem 19.2 in \cite{billingsley2012probability} for instance. For instance, if $\Omega$ is a Polish space, $\mathcal F$ is the Borel $\sigma$-algebra, and $\mathbb P$
is a Borel probability measure, then $(\mathcal{F},d_\mathcal{F})$ is separable.  
  For $L^\infty$, separability is more restrictive.  If $\Omega$ has only finitely many atoms (up to null sets), then $L^\infty$ is separable
(in fact finite-dimensional). Nonetheless, if $(\Omega,\mathcal F,\mathbb P)$ is atomless and standard (e.g.\ $([0,1],\mathcal B([0,1]))$ with Lebesgue measure),
then $L^\infty$ is not separable. Nonetheless, under separability of $L^\infty$, from \Cref{Thm:main}, we have that $(Ent_\gamma)_{\mu,r}$ also is Gâteaux differentiable for any $X\in L^\infty$ and can be represented by \[Q\mapsto\inf\limits_{\substack{\int_{\bar{B}(0,r)}Q_Zd\mu_{0,r}(Z)=Q \\ Z\mapsto\alpha_{{ENT_\gamma}}\left(Q_Z\right)-E_{Q_Z}[Z]\in\mathcal{B}}}
		\int_{\bar{B}(0,r)}\left(\frac{1}{\gamma}E\left[\frac{dQ_Z}{d\mathbb{P}}\log\frac{dQ_Z}{d\mathbb{P}}\right]-E_{Q_Z}[Z]\right)d\mu_{0,r}(Z).\]    
    Moreover, by \Cref{crl:subd},  if the supremum in \eqref{eq:dualmu} is attained we have that  \[\partial(Ent_\gamma)_{\mu,r}(X)=\int_{\bar{B}(0,r)}\frac{e^{-\gamma (X-Z)}}{E[e^{-\gamma( X-Z)}]}d\mu_{0,r}(Z).\]

\end{Exm}

\begin{Exm}[$\mathcal{X}$ as Orlicz spaces] \label{ex:orlicz}
Let $(\Omega,\mathcal{F},\mathbb{P})$ be a probability space and let
$\Phi:[0,\infty)\to[0,\infty)$ be a Young function, i.e.\ $\Phi$ is convex,
$\Phi(0)=0$, $\Phi(1)=1$, and $\Phi(x)\uparrow\infty$ as $x\to\infty$.
Let $\Psi$ denote the convex conjugate of $\Phi$,
$\Psi(y):=\sup_{x\ge 0}\{xy-\Phi(x)\},\:y\ge 0$.
The associated Orlicz space and Orlicz heart are, respectively,
\[
L^\Phi:=\Big\{X:\ \mathbb{E}\Big[\Phi\Big(\frac{|X|}{a}\Big)\Big]<\infty
\text{ for some }a>0\Big\},\;
M^\Phi:=\Big\{X:\ \mathbb{E}\Big[\Phi\Big(\frac{|X|}{a}\Big)\Big]<\infty
\:\forall\:a>0\Big\}.
\]
The Luxemburg norm on $L^\Phi$ is given by
\[
\|X\|_\Phi:=\inf\Big\{a>0:\ \mathbb{E}\Big[\Phi\Big(\frac{|X|}{a}\Big)\Big]\le 1\Big\}.
\]
Then $(L^\Phi,\|\cdot\|_\Phi)$ is a Banach lattice under the $\mathbb{P}-a.s.$ pointwise order,
and $(M^\Phi,\|\cdot\|_\Phi)$ is a (closed) Banach sublattice. If $\Phi$ satisfies
the $\Delta_2$-condition, i.e. if there exist constants $K>0$ and $x_0\ge 0$ such that
$\Phi(2x)\le K\,\Phi(x)\:\forall\:x\ge x_0$, then $M^\Phi=L^\Phi$. For instance, $L^p$ spaces, $p\in[1,\infty)$, are Orlicz spaces under the Young function
$\Phi_p:[0,\infty)\to[0,\infty)$ defined by $\Phi_p(t):=t^p/p$. The normed dual of $(M^\Phi,\|\cdot\|_\Phi)$ can be identified with $(L^\Psi,\|\cdot\|_\Psi^\ast)$,
where the Orlicz (dual) norm is
\[
\|Y\|_\Psi^\ast := \sup_{\|X\|_\Phi\le 1}\mathbb{E}[|XY|].
\]
Risk measures on $M^\Phi$ are then studied in, for instance, \cite{Cheridito2008}, \cite{Bellini2008}, \cite{Cheridito2009}. In this case we have, under the growth condition $\alpha_\rho(Q)\geq a+b\|Q\|^*_\Psi$ with $a,b\in\R,\:b>0$, that $\mathcal{X}^\prime_{1,+}=\{Q\in L^\Psi\colon Q\geq 0,\:E[Q]=1\}$, representing Radon-Nikodym derivatives with respect to $\mathbb{P}$.

The most known examples of risk measure in this space are the cash-invariant hulls. Let  $V:M^\Phi\to\R$ be convex and monotone (in the sense of risk measures). Then \[\rho_V(X):=\inf_{s\in\mathbb{R}}\{V(s-X)-s\}\]
is a real-valued convex risk measure on $M^\Phi$. Under mild assumptions, the infimum in this definition is attained. In order to make the averaging robust more tractable to this class of risk measures, we need some additional structure. We assume the linear space of measurable real valued maps $\mathcal{L}(\mathcal{X})$ on $(\mathcal{X},\mathcal{B})$ is decomposable, i.e. for any $Y\in \mathcal{L}(\mathcal{X})$, any $B\in\mathcal{B}$, and any
bounded $\mathcal{B}$-measurable $W\colon\mathcal{X}\to\R$, we have
$1_{B^c}Y+1_B W\in\mathcal{L}(\mathcal{X})$. Still, let $h:\mathbb{R}\times Z\to\mathbb{R}$ be defined as $h(s,Z)=s+V(s-Z)$, where $V$ is as above. It is straightforward to verify it is a (measurable) normal integrand, i.e.  its epigraphical mapping is closed valued and measurable. Notice that
$\inf_{S\in \mathcal{L}(\mathcal{X})}\,\int_{\bar{B}(0,r)}h(S(Z),Z)d\mu_{0,r}(Z)<\infty$.
Then, the interchange identity, see Theorem 14.60 in \cite{Rockafellar2009} for instance, holds:
\[
(\rho_V)_{\mu,r}(X)=\int_{\bar{B}(0,r)}\inf\limits_{s\in \R}\{V(s-X)-s\}\,d\mu_{0,r}(Z)
=
\inf_{S\in \mathcal{L}(\mathcal{X})}\,\int_{\bar{B}(0,r)}\left(S(Z)-V(S(Z)-Z)\right).\]
Hence, $\rho_V$ is integrable in closed balls and suitable for our approach. Moreover, since the common value is finite, then 
$S^*\in \mathcal{L}(\mathcal{X})$ attains the infimum in right-hand side if and only if $S(Z)\in\argmin_{s\in\mathbb{R}} h(s,Z)
\:\mu_{0,r}-a.s.$

Specializing, we have the class of transformed norm risk measures. Let
 $F:[0,\infty)\to(-\infty,\infty]$ be left-continuous, increasing and convex with
$\lim_{x\to\infty}F(x)=\infty$, $G$ a real-valued Young function, and
$H:\mathbb{R}\to[0,\infty)$ increasing and convex with $\lim_{x\to\infty}H(x)=\infty$.
Then, define $V(X):=F(\|H(X)\|_G)$, where $\|\cdot\|_G$ is the Luxemburg norm on the Orlicz heart $M^G$.
Then the cash-additive hull takes the explicit form $\rho(X)=\min_{s\in\mathbb{R}}\{F(\|H(s-X)\|_G)-s\}$. A particular case is the concept of Orlicz premia and Haezendonck--Goovaerts risk measure. Fix $\alpha\in(0,1)$ and set $\Phi_\alpha := \Phi/(1-\alpha)$ with convex conjugate $\Psi_\alpha$.
Define the Orlicz premium functional
\[
\pi_\alpha(X):=\inf_{x\in\mathbb{R}}\Big\{x+\|(X-x)^+\|_{\Phi_\alpha}\Big\},
\]
and the Haezendonck--Goovaerts risk measure by
$\rho_\alpha(X):=\pi_\alpha(-X),\:X\in M^\Phi$.
This is precisely the transformed norm risk measure corresponding to the choice $F(x)=x$, $G=\Phi_\alpha$, and $H(x)=x^+$, leading to a coherent risk measure. We have that a Orlicz heart is separable when $(\mathcal{F},d_\mathcal{F})$ is separable, see \cite{rao_ren_1991_orlicz} for instance. Under this property, from \Cref{Thm:main} and \Cref{rmk:coh}, we have that
\[
\alpha_{\mu,r}(Q)=
\inf\limits_{\substack{\int_{\bar{B}(0,r)}Q_Z\,d\mu_{0,r}(Z)=Q \\ \|Q\|^*_{\Psi_\alpha}\leq 1 \\Z\mapsto \langle Z,Q_Z\rangle\in\mathcal{B}}}
\int_{\bar{B}(0,r)}\langle -Z,Q_Z\rangle\,d\mu_{0,r}(Z).
\]
\end{Exm}

\begin{Exm}[$\mathcal{X}=C_b(\Omega)$]
 Let $\Omega$ be a Polish space, with $\mathcal{F}=\mathcal{B}(\Omega)$, i.e. the Borel sigma-algebra. We then restrict the domain of $\rho$ to be $\mathcal{X}=C_b:=C_b(\Omega)$, the space of continuous (and bounded) real valued functions on $\Omega$. This is a Banach lattice. In this case we have that $\mathcal{X}^\prime_{1,+}$ are the regular probability measures on $(\Omega,\mathcal{F})$. In \cite{Delbaen2022} and \cite{Nendel2024}, for instance, we have theory for risk measures in such a space. Every monetary risk measure in this space is Lipschitz continuous. Then, by \Cref{thm:r}, we have that $r\mapsto\rho_{\mu,r}(X)$ is continuous and  $\lim\limits_{r\to 0}\rho_{\mu,r}(X)=\rho(X)$ for any $X\in C_b$.
 
 Fix a reference probability measure $\mathbb{P}$ on $(\Omega,\mathcal{F})$ and define $\mathcal{X}^\prime_\mathbb{P}=\{Q\in\mathcal{X}^\prime_{1,+}\colon Q\ll\mathbb{P}\}$. Let again $E$ be the expectation with respect to $\mathbb{P}$. Now, let a divergence be a map $Q\mapsto E[g(Q)]$, where $g\colon \mathbb{R}_+\to \mathbb{R}_+$ is convex, lsc and $E[g(Q)]=\infty$ for any $Q\not\in\mathcal{X}^\prime_\mathbb{P}$. In this case we can define a divergence based risk measure $\rho_g$ in $C_b$ as \[\rho_g(X)=\max\limits_{Q\in\mathcal{X}^\prime_{1,+}}\{E_Q[-X]-E[g(Q)]\}.\] It is clearly a convex risk measure with $\alpha_\rho=g$. Examples are the Entropic, Shortfall Risks as in \cite{Follmer2016}, and any coherent risk measure (since the penalty term is only a convex indicator function). It is a well known fact that $C_b(\Omega)$ is separable if and only if $\Omega$ is compact. In this case, $C_b(\Omega)=C(\Omega)$. By making a slightly abuse of notation by treating $Q\in\mathcal{X}^\prime_\mathbb{P}$ by its Radon-Nikodym derivative, we then can compute $\alpha_{\mu,r}$ by considering a measurable $q:\bar{B}(0,r)\times\Omega\to[0,\infty)$ satisfying the marginal constraints
\begin{equation*}
\int_{\bar{B}(0,r)} q(Z,\omega)\,\mu_{0,r}(dZ)=Q(\omega)\quad\text{for $P$-a.s. }\omega,
\qquad
\int_\Omega q(Z,\omega)\,d\mathbb{P}(\omega)=1\quad\text{for $\mu_{0,r}$-a.s. }Z.
\end{equation*} We have that $(Z,\omega)\to Z(\omega)$ is measurable in the product space, since it is continuous on $(Z,\Omega)$ for the product  metric as $Z\mapsto Z(\omega)$ is continuous in the supremum norm. Then, by \Cref{Thm:main} we have that
\begin{align*}
  \alpha_{\mu,r}(Q) := & \inf \int_{\bar{B}(0,r)\times\Omega} \left(g(q(Z,\omega))-Z(\omega)q(Z,\omega)\right)\,d(\mu_{0,r}\times \mathbb P)(Z,\omega) \\
  &\text{s.t.}\\
  &\int_{\bar{B}(0,r)}q(Z,\omega)\,d\mu_{0,r}(Z)=Q(\omega)\quad\text{for $P$-a.s. }\omega,\\
  & \int_\Omega q(Z,\omega)\,d\mathbb{P}(\omega)=1\quad\text{for $\mu_{0,r}$-a.s.}Z,\\
  &q\geq 0.
\end{align*}
In the particular case where $E[g(Q)]=E_Q[h(Q)]$, as is the case for the Entropic risk measure and any coherent one, this penalty term computation is a classical optimal transport problem. Let $\nu_\Omega:=Qd\mathbb P$ and $\nu_{\bar{B}(0,r)}:=\mu_{0,r}$. Both are probability measures. The constraints mean that $q$ is a coupling of $\nu_\Omega$ and $\nu_{\bar{B}(0,r)}$:
\[
\Pi\!\left(\nu_\Omega,\nu_{\bar{B}(0,r)}\right)
:=\Big\{q\ge 0:\ \int_{\bar{B}(0,r)} q\,d\mu_{0,r}=Q\ \mathbb{P}\text{-a.s.},\ \int_\Omega q\,d\mathbb P=1\ \mu_{0,r}\text{-a.s.} \Big\}.
\] In this context we get \[\alpha_{\mu,r}(Q)
= \inf_{q\in\Pi(\nu_\Omega,\nu_{\bar{B}(0,r)})} \int_{\bar{B}(0,r)\times\Omega}\left(h(q(Z,\omega))-Z(\omega)q(Z,\omega)\right)\,d(\mu_{0,r}\times \mathbb P).\]
This problem can be solved through Kantorovich duality (see \cite{Villani2021}) with $(\omega,Z)\mapsto h(q(Z,\omega))-Z(\omega)$ as  a valid cost function. 
\end{Exm}

\section{Aggregate then evaluate}\label{sec:agg}

 To deal with distribution functions, we fix a probability measure $\mathbb{P}$ in $(\Omega,\mathcal{F})$. Let $F_{X}(x) = \mathbb{P}(X\leq x)$ and $F_{X}^{-1}(\alpha)=\inf\{x\in\mathbb{R}\colon F_X (x)\geq\alpha\}$ for $\alpha \in (0,1)$ be, respectively, the distribution function and the (left) quantile of $X$. We also let $E$ be the expectation operator. In this context, we are interested in functionals that satisfy the following property:
\begin{itemize}
   \itemsep0em 
    \item[] Law Invariance: if $F_X=F_Y$, then $\rho(X)=\rho(Y),\:\forall\:X,Y\in \mathcal{X}$.
\end{itemize}
We call risk measures with this property law invariant. Then, with a slight abuse of notation, we write $\rho(F_X)$ in place of $\rho(X)$,
and view $\rho$ as a real-valued functional on the set $\mathcal{M}$ of distribution functions. We also have the subset $\mathcal{M}(\mathcal{X}):=\{F_X\colon\:X\in \mathcal{X}\}$. Under Law Invariance, we always assume $(\Omega,\mathcal{F},\mathbb{P})$ is atomless, in which case it supports a random variable $U$ that is uniformly distributed over $[0,1]$. We call $\rho$ a spectral risk measure if \[\rho_\phi(X)=-\int_0^1F^{-1}_X(u)\phi(u)du,\:\forall\:X \in \mathcal{X},\] where $\phi:[0,1]\to\mathbb{R}^+$ is a non-increasing functional such that $\int_{0}^{1}\phi(u)du=1$. The most prominent example of spectral risk measure is the Expected Shortfall (ES), that is functional $ES_\alpha\colon L^1\to\mathbb{R}$ defined as \[ES_\alpha(X)=-\frac{1}{\alpha}\int_0^\alpha F^{-1}_X(u)du,\:\alpha\in(0,1).\]In this case the spectral function is $\phi(u)=\frac{1}{\alpha}1_{(0,\alpha)}(u),\:\alpha\in(0,1)$. 

\begin{Def} \label{def:diff_aggregations}
For any $X\in\mathcal{X}$ and $r>0$, we define $X_{\mu,r}:=F^{-1}_{\mu,r}(U)$ and  $X_{\bar{\mu},r}:=F^{-1}_{\bar{\mu},r}(U)$, where \begin{equation}\label{eq:fmu}F_{\mu,r}(x)=\int_{\bar{B}(X,r)}F_Z(y)d\mu_{X,r}(Z),\end{equation} and \begin{equation}\label{eq:fmuq}F^{-1}_{\bar{\mu},r}(u)=\int_{\bar{B}(X,r)}F^{-1}_Z(u)d\mu_{X,r}(Z).\end{equation}
\end{Def}

We now give sufficient conditions for both aggregations to be well defined. Notice that the relation $X_{\bar{\mu},r}=X_{\mu,r}$ is not always true.

\begin{Lmm}\label{lmm:agg}
 The integral in \eqref{eq:fmu} (respectively, \eqref{eq:fmuq}) is proper for any $x\in\mathbb{R}$ (for any $u\in(0,1)$) if and only if $X_{\mu,r}$ (respectively, $X_{\bar{\mu},r}$) is a random variable.
\end{Lmm}

\begin{proof}
The if part is straightforward. For the only if step, by monotonicity of the integral, we get that $x\mapsto F_{\mu,r}(x)$ is monotone and uniformly bounded, with respect to $Z\in\bar{B}(X,r)$, since it lies on $[0,1]$. Thus, by the Dominated Convergence Theorem it is right-continuous, hence a valid distribution. Then, $X_{\mu,r}:=F_{\mu,r}^{-1}(U)$ is well defined. Now, we must prove that $F^{-1}_{\bar{\mu},r}$ is a proper quantile function. By monotonicity of the integral, we have that $v\colon u\mapsto F_{\bar{\mu},r}^{-1}(u)$ is monotone. Moreover, by Fatou's Lemma, since the sequence of maps $Z\mapsto F^{-1}_Z(u_n)$ is bounded below for sufficiently large $n$ for any convergent $u_n\to u$, $v$ is lower semi continuous in $(0,1)$. Hence, $u\mapsto F_{\bar{\mu},r}^{-1}(u)$ is a valid quantile function. Thus, $X_{\bar{\mu},r}:=F_{\bar{\mu},r}^{-1}(U)$ is well defined.
\end{proof}

The next main result claims a relationship between the aggregation approach and the one we propose.

\begin{Thm}\label{thm:avg}
 If $X_{\mu,r}\in\mathcal{X}$ and  $\rho$ is concave on $\mathcal{M}(\mathcal{X})$, then $\rho(X_{\mu,r})\geq\rho_{\mu,r}(X)$. Moreover, if $X_{\bar{\mu},r}\in\mathcal{X}$ and $\rho$ is a spectral risk measure, then $\rho(X_{\bar{\mu},r})=\rho_{\mu,r}(X)$. In this case, $\rho_{\mu,r}(0)=0$ (normalized) if and only if $\rho=-E$.
\end{Thm}

\begin{proof}
Let $\nu_Z$ be the measure associated with $F_Z$, $\nu_Z((a,b]):= F_Z(b)-F_Z(a)$, and define the barycentric measure
\[
\nu := \int_{\bar{B}(X,r)} \nu_Z\,d\mu_{X,r}(Z)
\]
as a Pettis integral in the space of Borel probability measures $\mathcal{M}(\mathbb{R})$. This space can be identified with the topological dual of $C_0(\mathbb{R})$, the space of continuous real functions that vanishes at infinity. For a reference for this duality, known as the Riesz–Markov-Kakutani Theorem, see Theorem 6.19 in  \cite{rudin1987real} for instance. For each rational $q$, the map $Z\mapsto F_Z(q)=\nu_Z((-\infty,q])$ is Borel measurable on $\mathcal{X}$. Hence, $Z\mapsto \nu_Z$ is weak$^*$-measurable in $C_0(\mathbb{R})^*$ and the Pettis (weak$^*$) integral
$\int \nu_Z\,d\mu_{X,r}(Z)$ is well-defined.
For every $x\in\mathbb{R}$ we then have that 
\[
\nu((-\infty,x]) 
  = \int_{\bar{B}(X,r)} \nu_Z((-\infty,x])\,d\mu_{X,r}(Z)
  = \int_{\bar{B}(X,r)} F_Z(x)d\mu_{X,r}(Z).
\]
Since $\nu$ is a Borel measure, its cumulative distribution function $F_\nu(x) := \nu((-\infty,x])$
is right-continuous.  If $d$ is a metric on $\mathcal{M}$, then $F_X=F_Y$ implies $\bar{B}(X,r)=\bar{B}(Y,r)$. In this case, we have that Law Invariance of $\rho$ is then preserved by $\rho_{\mu,r}$. Therefore, $F_{\mu,r}=F_\nu$, representing the Pettis (or barycentric) measure $\nu$. By Theorem 3 in \cite{vesely2017jensen}, we then have the following Jensen inequality for the Pettis integral \[\rho_{\mu,r}(X)=\int_{\bar{B}(X,r)}\rho(F_Z)d\mu_{X,r}(Z)\leq \rho(F_\nu)=\rho(F_{\mu,r})=\rho(X_{\mu,r}).\]

For the second claim, since $(u,Z)\mapsto F_{Z}^{-1}(u)$ is integrable on the product space, we can use the Fubini-Tonelli Theorem in order to get that \[\rho_{\mu,r}(X)=-\int_0^1\int_{\bar{B}(X,r)}F^{-1}_{Z}(u)d\mu_{X,r}(Z)\phi(u)du=\rho(X_{\bar{\mu},r}).\] Moreover,  since $-F^{-1}_{Z}(u)=F^{-1}_{-Z}(1-u)$, we have that $\rho_{\mu,r}(0)=0$ if and only if $\phi(u)=\phi(1-u)$. Since $\phi$ is monotone, this implies that we have normalization if and only if $\phi=1$, where we recover the expectation. This concludes the proof.
\end{proof}

We have the following corollary, that compares the robust averaging risk measure $\rho_{\mu,r}$ with its aggregated counterparts.  The notation $X\succeq Y$, for $X,Y\in \mathcal{X}$, indicates second-order stochastic dominance, that is, $E[f(X)]\leq E[f(Y)]$ for any increasing convex function $f\colon\mathbb{R}\rightarrow\mathbb{R}$.

\begin{Crl}\label{crl:avg}
Let $\mathcal{X}\subseteq L^1$ be a separable Banach lattice, $\rho$ a spectral risk measure, and \Cref{asp:lin} holds. Then  \[\rho(X)\leq \rho(X_{\bar{\mu},r})=\rho_{\mu,r}(X)\leq \rho(X_{\mu,r}).\] In particular we have that $X\succeq X_{\bar{\mu},r}\succeq X_{\mu,r}$. 
\end{Crl}

\begin{proof}
By \Cref{Thm:main}, we already have $\rho(X)\leq\rho_{\mu,r}(X)$.  Moreover, by \Cref{thm:avg} and \Cref{ex:lp}, we have that $\rho_{\mu,r}(X)=\rho(X_{\bar{\mu},r})\leq\rho(X_{\mu,r})$. The second-order stochastic dominance $X\succeq Y$ is equivalent to $ES^\alpha(X)\leq ES^\alpha(Y)$ for any $\alpha\in[0,1]$, see \cite{Bauerle2006} for instance. Thus, the claim follows directly from \Cref{thm:avg}.
\end{proof}

\begin{Exm}[$L^p$ and Law Invariance]
Regarding the framework of Law Invariance on $L^p$ spaces, \cite{Filipovic2012} is a reference, for instance. In fact, every law invariant comonotone (additive for comonotone pairs) convex risk measure in $L^p,\:p\in[1,\infty]$ is a spectral one. Moreover, we have that $Z\mapsto F_Z(x)$ is integrable for any $x\in\mathbb{R}$. For measurability of $Z\mapsto F_Z(x)$, choose bounded Lipschitz functions $\phi_n:\mathbb R\to[0,1]$ such that
$\phi_n\downarrow 1_{(-\infty,x]}$ point-wise.
Then $Z\mapsto \mathbb E[\phi_n(Z)]$ is continuous on $L^p$.
By monotone convergence, we get $F_Z(x)=\mathbb E[1_{(-\infty,x]}(Z)]=\lim_{n\to\infty}\mathbb E[\phi_n(Z)]$. Hence, $Z\mapsto F_Z(x)$ is Borel measurable on $L^p$. Since distributions lie on the unit interval, we have boundedness and, thus, integrability. For quantiles, Borel-measurability follows by using the definition of the generalized inverse since
\[
\{Z:\ F_Z^{-1}(u)<a\}=\bigcup_{q\:\text{rational},\ q<a}\{Z:\ F_Z(q)\ge u\},
\]
which is Borel measurable since $Z\mapsto F_Z(x)$ also is. Thus $Z\mapsto F_Z^{-1}(u)$ is Borel measurable on $L^p$. Moreover, for any random variable \(Z\) with 
\(\|Z\|_p \le R:=\|X\|_p+r\), by Markov's inequality,
\[
\mathbb{P}(|Z| \ge t) 
\le \frac{\mathbb{E}|Z|^p}{t^p} 
\le \frac{R^p}{t^p}.
\] Thus, it is straightforward to verify that \(Z \mapsto F_Z^{-1}(u)\) is bounded on the set 
    \(\{Z : \|Z\|_p \le R\}\) for any $u\in(0,1)$. Therefore,  $F^{-1}_{\bar{\mu},r}$ it is well defined. An application of the Jensen inequality together to such integrability implies that  $X_{\mu,r},X_{\bar{\mu},r}\in L^p$. In particular, spectral risk measures are integrable over closed balls, naturally fitting in our framework for $L^p$ spaces.
A sufficient condition for $R_\rho$ be concave in distributions when $\mathcal{X}=L^p$ is, see \cite{Acciaio2013} for details, $\rho$ is a law invariant convex risk measure, and $R_\rho^*$ is concave, where \[R_\rho^*(F)=\alpha_\rho(Q),\:Q\sim F.\] For example, when $\rho=Ent_\gamma$, then \[R_\rho^*(F)=\gamma^{-1}\int_{-\infty}^\infty x\log xdF(x),\] which is linear. Thus, the Entropic risk measure is concave in distributions and fits in the scope of \Cref{thm:avg}. This also holds for spectral risk measures.

    We can think of a ``median" like aggregation of quantiles. Let $W\in L^2$, the space of random variables with finite second moment, which is a Banach lattice.
Write  \(\sigma(W)\) for the standard-deviation of $W$. Cantelli’s one–sided inequality gives, for any real r.v. \(W\),
\[
F^{-1}_W(u)\le m(W)+\sigma(W)\sqrt{\tfrac{u}{1-u}},\qquad
F^{-1}_W(u)\ge m(W)-\sigma(W)\sqrt{\tfrac{1-u}{u}}.
\]

Apply this to \(W=Y_Z=X+Z\). Then
\[
m(Y_Z)=E[X]+m(Z),\qquad 
\sigma(Y_Z)\le \sigma(X)+\sigma(Z)\le \sigma(X)+2r,
\]
since \(\sigma(Z)\le |m(Z)|+ \|Z\|_2\le 2r\) for any $Z\in\bar{B}(0,r)$. Linearity and \Cref{lmm:r} give
\[
\int_{\bar{B}(0,r)} m(Y_Z)\,d\mu(Z)=E[X].
\]
Using the uniform bound \(\sigma(Y_Z)\le \sigma(X)+2r\) inside Cantelli’s inequalities yields the clean two–sided estimate for $u\in(0,1)$:
\[\int_{\bar{B}(0,r)} F^{-1}_{X+Z}(u)\,d\mu_{0,r}(Z)\in\left[
E[X]-\big(\sigma(X)+2r\big)\sqrt{\frac{1-u}{u}},
E[X]+\big(\sigma(X)+2r\big)\sqrt{\frac{u}{1-u}}\right]. 
\]
A natural minimax point estimate is the midpoint of this interval:
\[
\widehat{F^{-1}_{\bar{\mu},r}}(u)
:= E[X]+\big(\sigma(X)+2r\big)\,\frac{2u-1}{2\sqrt{u(1-u)}}.\]
\end{Exm}

\begin{Exm}[Wasserstein metric] \label{ex:wass}
Fix $p\in[1,\infty]$ and let $d=d_{W_p}$ be the one-dimensional Wasserstein metric of order $p$,
defined through quantiles by
\[
d_{W_p}(X,Z)
:=\|F_X^{-1}-F_Z^{-1}\|_{L^p(0,1)}
=
\begin{cases}
\Big(\int_0^1 |F_X^{-1}(u)-F_Z^{-1}(u)|^p\,du\Big)^{1/p}, & p<\infty,\\[4pt]
\operatorname*{ess\,sup}_{u\in(0,1)} |F_X^{-1}(u)-F_Z^{-1}(u)|, & p=\infty.
\end{cases}
\]
For a detailed discussion on this metric, see \cite{Villani2021}, while \cite{Esfahani2018} is a reference for its use in robust decision-making. For some sensitivity analysis of risk measures in this context, see \cite{Bartl2021} and \cite{Nendel2022}. It is straightforward to verify that $\bar{B}_{L^p}(X,r)\subseteq \bar{B}_{W^p}(X,r)$. Nonetheless, they agree when the center of the ball is a constant. 

This Wasserstein-based example is meant to illustrate the averaging construction in a metric
space of laws. Since Wasserstein balls need not coincide with norm balls in a Banach lattice,
the functional-analytic assumptions used in \Cref{thm:r} do not automatically apply in this setting unless explicitly imposed. Let $(\Omega,\mathcal F,\mathbb P)$ be atomless. We assume here the paradigm of Law Invariance, i.e. we can consider maps $X\mapsto\rho(F_X)$. Fix a center $X\in \mathcal{X}$ and write $q_X:=F_X^{-1}$.
Choose $h\in L^p(0,1)$ such that:
(i) $h$ is nondecreasing (a.e.); (ii) $\|h\|_{L^p}=1$; and (iii) for every $a\in[-r,r]$,
the function $q_X+a h$ is (a.e.) non-decreasing, hence it is a quantile function.
For each $a\in[-r,r]$, define $Z_a\in \mathcal{X}$ by prescribing its quantile
$q_{Z_a}(u):=F_{Z_a}^{-1}(u)=q_X(u)+a\,h(u),\; u\in(0,1)$.
Then \[d_{W_p}(X,Z_a)=\|q_X-q_{Z_a}\|_{L^p}=\|a h\|_{L^p}=|a|,\] so $\{Z_a:\ |a|\le r\}\subseteq\bar B(X,r)$. Let $\gamma$ be supported on the curve $\{Z_a:\ |a|\le r\}$ (equivalently, if we evaluate the
aggregation along this one-direction family). Then, define $I:Z_a\mapsto a$ and set $\eta:=I_\#\gamma_X$.
Then the induced local law on $a$ is
\[
\eta_{X,r}(da)
:=I_\#\mu_{X,r}(da)
=\frac{1}{\widetilde K(X,r)}\,1_{\{|a|\le r\}}\,\varphi(|a|)\,\eta(da),
\qquad
\widetilde K(X,r):=\int_{-r}^r \varphi(|a|)\,\eta(da).
\]
Hence, for every $u\in(0,1)$ we obtain the explicit representation
\[
F_{\bar\mu,r}^{-1}(u)
=\int_{-r}^r \big(q_X(u)+a\,h(u)\big)\,\eta_{X,r}(da)
= q_X(u)+m_{X,r}\,h(u),
\]
where
\[
m_{X,r}:=\int_{-r}^r a\,\eta_{X,r}(da)
=\frac{\int_{-r}^r a\,\varphi(|a|)\,\eta(da)}{\int_{-r}^r \varphi(|a|)\,\eta(da)}.
\]
In particular, since $\eta_{X,r}$ is supported on $[-r,r]$, we have $|m_{X,r}|\le r$ and thus
$\|F_{\bar\mu,r}^{-1}-F_X^{-1}\|_{L^p}=\|m_{X,r}h\|_{L^p}=|m_{X,r}|\le r$.

Using the explicit quantile formula above we obtain, by \Cref{thm:avg}, for any spectral risk measure
\[
\rho_\phi(\bar X_{\mu,r})
=
-\int_0^1 \big(q_X(u)+m_{X,r}h(u)\big)\phi(u)\,du
=
\rho_\phi(X)
-
m_{X,r}\int_0^1 h(u)\phi(u)\,du.
\]
Consequently,
\[
\big|\rho_\phi(\bar X_{\mu,r})-\rho_\phi(X)\big|
\le |m_{X,r}|\int_0^1 |h(u)|\phi(u)\,du
\le r\int_0^1 |h(u)|\phi(u)\,du.
\]
Notice that the inequality
$\rho_\phi(X)\le \rho_{\mu,r}(X)$ is not assured since the sign depends on the product
$m_{X,r}\cdot\int_0^1 h\phi$. This is not a contradiction to \Cref{Thm:main} since in the Wasserstein setting of this example the ball is defined by the
metric $d_{W_p}$, which is not a norm on a Banach lattice.
\end{Exm}

\section{The Gaussian measure and numerical examples} \label{sec:gaussian_examples}

In this section we discuss several applications of the results presented throughout the paper. Ubiquitous to the Gaussian examples is the non-central chi-squared distribution, which is formally defined below.

\begin{Def}(The non-central Chi-Squared Distribution) Let $X_1, \ldots,X_n$ be a sequence of independent random variables where $X_k \sim N(\mu_k, 1)$, for $k=1,\ldots,n$. The distribution of $\sum_{k=1}^n X_k^2$ is the non-central chi-square distribution with $n$ degrees of freedom and non-centrality parameter $\lambda = \sum_{k=1}^n \mu_k^2$.
\end{Def}

A standard result of these distributions (see \cite{fisher1928general} or \cite{siegrist2022probability} for a more modern derivation) states the density of a non-central chi-square (and hence its cdf) can be written as an infinite weighted sum of centered chi-square densities with weights given by Poisson pdf's.

\begin{Prp}[Poisson mixture representation] \label{prop:poisson_mixture}
If $J\sim\mathrm{Poisson}(\lambda/2)$ independently of $\chi^2_{k+2J}$, then
$\chi^2_k(\lambda)\overset{d}{=}\chi^2_{k+2J}$, so
\begin{equation}\label{eq:mixture}
  F_k(x;\lambda)
  \;=\;
  \sum_{j=0}^{\infty} p_j(\lambda)\,F_{\chi^2_{k+2j}}(x),
  \qquad
  p_j(\lambda) \;:=\; e^{-\lambda/2}\frac{(\lambda/2)^j}{j!},
\end{equation}
where $F_{\chi^2_{k+2j}}(x)$ denotes the cdf of the \emph{central}
chi-squared distribution with $k+2j$ degrees of freedom.  Moreover, the sum in \eqref{eq:mixture} converges uniformly on compact sets in $\lambda$.
\end{Prp}

As an almost immediate consequence of the Poisson mixture representation, one can derive a simple identity for the derivative of the non-central chi-squared distribution.

\begin{Crl} The derivative of the non-central chi-squared cdf can be written as
\begin{equation}\label{eq:cdf-deriv}
  \frac{\partial F_k(x;\lambda)}{\partial\lambda}
  \;=\;
  \tfrac{1}{2}\!\left[F_{k+2}(x;\lambda) - F_k(x;\lambda)\right].
\end{equation}
\end{Crl}

\begin{proof}
We prove this using the Poisson mixture representation of Proposition \ref{prop:poisson_mixture}. As the sum converges uniformly on compact sets in $\lambda$, we may differentiate term by term.

A direct calculation gives
\[
  \frac{\partial\, p_j(\lambda)}{\partial\lambda}
  \;=\;
  -\frac{1}{2}e^{-\lambda/2}\frac{(\lambda/2)^j}{j!}
  \;+\;
  e^{-\lambda/2}\frac{(\lambda/2)^{j-1}}{2[(j-1)!]}\cdot 1_{\{j\geq 1\}}
  \;=\;
  \tfrac{1}{2}\left[p_{j-1}(\lambda) - p_j(\lambda)\right],
\]
with the convention $p_{-1}\equiv 0$.  Therefore
\[
  \frac{\partial F_k(x;\lambda)}{\partial\lambda}
  \;=\;
  \frac{1}{2}
  \sum_{j=0}^{\infty}
  \left[p_{j-1}(\lambda) - p_j(\lambda)\right] F_{\chi^2_{k+2j}}(x).
\]
 
Splitting the sum and applying the shift $j\mapsto j+1$ to the first part
(the $j=0$ term vanishes since $p_{-1}=0$):
\[
  \sum_{j=0}^{\infty} p_{j-1}(\lambda)\,F_{\chi^2_{k+2j}}(x)
  \;=\;
  \sum_{j=0}^{\infty} p_j(\lambda)\,F_{\chi^2_{k+2(j+1)}}(x)
  \;=\;
  \sum_{j=0}^{\infty} p_j(\lambda)\,F_{\chi^2_{(k+2)+2j}}(x)
  \;=\;
  F_{k+2}(x;\lambda),
\]
where the last equality uses the representation \eqref{eq:mixture} with $k+2$
in place of $k$.  Hence
\[
  \frac{\partial F_k(x;\lambda)}{\partial\lambda}
  \;=\;
  \tfrac{1}{2}\!\left[F_{k+2}(x;\lambda) - F_k(x;\lambda)\right],
\]
which establishes \eqref{eq:cdf-deriv}.
\end{proof}

The following proposition is necessary to study the limit involved in the finite-dimensional characterization of the robust risk measure.

\begin{Prp}\label{prop:cdf-ratio}
For every fixed $x > 0$ and $\lambda \geq 0$,
\[
  \frac{F_{n+2}(x;\,\lambda)}{F_n(x;\,\lambda)}
  \;\sim\;
  \frac{x}{n}
  \qquad\text{as } n\to\infty.
\]
\end{Prp}

\begin{proof}
The proof proceeds in three steps. We first derive some asymptotic results for the central chi-squared cdf. Then, we study the convergence of ratios of consecutive central chi-squared cdfs finally of ratios of non-central chi-squared cdfs.

The cdf of $\chi^2_m$ satisfies
\[
  F_{\chi^2_m}(x)
  \;=\;
  \frac{\gamma(m/2,\;x/2)}{\Gamma(m/2)},
\]
where $\gamma(a,z) = \int_0^z t^{a-1}e^{-t}\,dt$ is the lower incomplete
gamma function.  We claim that for fixed $z = x/2 > 0$,
\begin{equation}\label{eq:gamma-asymp}
  \gamma(a,z) \;\sim\; \frac{z^a e^{-z}}{a}
  \qquad\text{as }a\to\infty,
\end{equation}
in the sense that the ratio of the two sides tends to $1$.
Substituting $t = z-s$:
\[
  \gamma(a,z)
  \;=\;
  e^{-z}z^{a-1}\int_0^z\!\left(1-\frac{s}{z}\right)^{a-1}\!e^{s}\,ds.
\]
Using $1-s/z \leq e^{-s/z}$ we obtain the upper bound
\[
  \int_0^z\!\left(1-\frac{s}{z}\right)^{a-1}\!e^{s}\,ds
  \;\leq\;
  \int_0^{\infty}e^{-s(a-1)/z}\,e^{s}\,ds
  \;=\;
  \int_0^{\infty}e^{-s(a/z - 1/z - 1)}\,ds
  \;=\;
  \frac{z}{a-1-z},
\]
valid for $a > 1+z$, so that
\[
  \gamma(a,z)
  \;\leq\;
  e^{-z}z^{a-1}\cdot\frac{z}{a-1-z}
  \;=\;
  \frac{z^a e^{-z}}{a-1-z}
  \;\sim\;
  \frac{z^a e^{-z}}{a}.
\]
For the lower bound, fix $\delta\in(0,z)$ and restrict to $s\in[0,\delta]$. Since $g(s) = \log(1-s/z) + s/(z-\delta)$ satisfies $g(0) = 0$ and $g'(s) \geq 0$ on $[0,\delta]$, we have that $\log(1-s/z)\geq -s/(z-\delta)$ for $s\leq\delta$. Hence,
\[
  \int_0^z\!\left(1-\frac{s}{z}\right)^{a-1}\!e^{s}\,ds
  \;\geq\;
  \int_0^{\delta}e^{-(a-1)s/(z-\delta)}\,ds
  \;=\;
  \frac{z-\delta}{a-1}\!\left(1-e^{-(a-1)\delta/(z-\delta)}\right).
\]
Putting everything together, the lower incomplete gamma function is bounded by
\[  e^{-z}z^{a-1}\frac{z-\delta}{a-1}\!\left(1-e^{-(a-1)\delta/(z-\delta)}\right) \leq \gamma(a,z) \leq \frac{z^a e^{-z}}{a-1-z}.\]
Dividing by $z^a e^{-z}/a$, we get that
\[\frac{z-\delta}{z}\frac{a}{a-1}\!\left(1-e^{-(a-1)\delta/(z-\delta)}\right) \leq \frac{\gamma(a,z)}{z^a e^{-z}/a} \leq \frac{a}{a-1-z}.\]
Taking $a\rightarrow \infty$ for fixed $\delta$ and then $\delta \rightarrow 0$, we get \eqref{eq:gamma-asymp}.  

Applying Stirling's approximation
$\Gamma(m/2+1)\sim\sqrt{\pi m}\,(m/2e)^{m/2}$ then yields
\begin{equation}\label{eq:central-asymp}
  F_{\chi^2_m}(x)
  \;\sim\;
  \frac{e^{-x/2}}{\sqrt{\pi m}}
  \left(\frac{ex}{m}\right)^{m/2}
  \qquad\text{as }m\to\infty.
\end{equation}

Forming the ratio $F_{\chi^2_{m+2}}(x)/F_{\chi^2_m}(x)$ using
\eqref{eq:central-asymp}:
\[
  \frac{F_{\chi^2_{m+2}}(x)}{F_{\chi^2_m}(x)}
  \;\sim\;
  \sqrt{\frac{m}{m+2}}
  \cdot
  \frac{ex}{m+2}
  \cdot
  \left(\frac{m}{m+2}\right)^{m/2}.
\]
Since $\sqrt{m/(m+2)}\to 1$ and 
\[(m/(m+2))^{m/2} = (1-2/(m+2))^{m/2} = \left[\left( 1 + \frac{(-1)}{\frac{m+2}{2}}\right)^{\frac{m+2}{2}} \right]^{\frac{m}{m+2}}\to
e^{-1},\]
we obtain
\begin{equation}\label{eq:central-ratio}
  \frac{F_{\chi^2_{m+2}}(x)}{F_{\chi^2_m}(x)}
  \;\sim\;
  \frac{x}{m}
  \qquad\text{as }m\to\infty.
\end{equation}
Iterating \eqref{eq:central-ratio} $j+1$ times starting from $m=n$:
\begin{equation}\label{eq:central-ratio-j}
  \frac{F_{\chi^2_{n+2j+2}}(x)}{F_{\chi^2_n}(x)}
  \;\sim\;
  \prod_{k=0}^{j}\frac{x}{n+2k}
  \;=\;
  O\!\left(\!\left(\frac{x}{n}\right)^{\!j+1}\right),
\end{equation}
which decays as $(x/n)^{j+1}\to 0$ for every fixed $j\geq 0$.

Using the Poisson mixture representation of Proposition \ref{prop:poisson_mixture} and factoring out $F_{\chi^2_n}(x)$ out of both the numerator and denominator:
\[
  \frac{F_{n+2}(x;\lambda)}{F_n(x;\lambda)}
  \;=\;
  \frac{\displaystyle\sum_{j=0}^{\infty}p_j(\lambda)\,
        \dfrac{F_{\chi^2_{n+2j+2}}(x)}{F_{\chi^2_n}(x)}}
       {\displaystyle\sum_{j=0}^{\infty}p_j(\lambda)\,
        \dfrac{F_{\chi^2_{n+2j}}(x)}{F_{\chi^2_n}(x)}}.
\]
By \eqref{eq:central-ratio-j}, the $j$-th term in the numerator is
$O((x/n)^{j+1})$ and in the denominator is $O((x/n)^j)$, so the $j=0$ term
dominates each sum:
\[
  \frac{F_{n+2}(x;\lambda)}{F_n(x;\lambda)}
  \;=\;
  \frac{p_0(\lambda)\cdot\dfrac{x}{n}(1+o(1)) + O\!\left(\dfrac{x^2}{n^2}\right)}
       {p_0(\lambda)\cdot 1 + O\!\left(\dfrac{x}{n}\right)}
  \;\sim\;
  \frac{x}{n},
\]
where $p_0(\lambda) = e^{-\lambda/2}$ cancels. 
\end{proof}

In the remainder of this section we specialize the choices of Example \ref{ex:gaus} and present some numerical results. The reader is reminded that $\mathcal{X}$ is a real, separable Hilbert space with orthonormal basis $\{e_k\}_{k\ge 1}$. In this space we let $\gamma = \mathcal{N}(0,C)$ be a centered Gaussian measure with covariance operator C with $Ce_k = \lambda_ke_k$.

\begin{Exm}[Gaussian measure with linear risk] \label{sec:gauss_linear} 
We first study the robust risk measure when the baseline risk is given by the negative of the first component of the random variable $Z$. More formally, we take $\varphi(t) = e^{-t^2}; \, X = e_1 = (1,0,0,\ldots); \, \rho(Z) = f(Z) = -Z_1$ and $\lambda_1=1$. As the integrand $f$ only depends on the first component of $Z$, the eigenvalues $(\lambda_k)_{k\geq 2}$ do not impact the results and are also set to 1, without loss of generality. Let us denote $u \sim \gamma_n = N(0,I_n)$. Our interest lies on computing
\begin{equation}
\rho_{\mu,r}(X)=\int_{\bar{B}(X,r)}\rho(Z)d\mu_{X,r}(Z) = \lim_{n\rightarrow\infty} \frac{E_{\gamma_n}[-u_1 e^{-\|u-X^{(n)} \|^2} 1_{\|u-X^{(n)}\| \leq r]}]}{E_{\gamma_n}[e^{-\|u-X^{(n)} \|^2}  1_{\|u-X^{(n)}\| \leq r}]}
\label{eq:ex1_limit}    
\end{equation}
using the finite-dimensional measures on the right hand side of \eqref{eq:ex1_limit}.

First, note that
\[\| u-X^{(n)}\|^2 = (u_1 - 1)^2 + \sum_{i=2}^n u_i^2\]
and that
\[e^{-\|u-X^{(n)} \|^2} = e^{-\|u\|^2} e^{2u_1 -1}.\] As $\gamma_n$ has standard Gaussian density in $\R^n$ by combining the exponentials, we have that
\[E_{\gamma_n}[g(u) e^{\|u-X^{(n)}\|^2}] \propto E_{\tilde \gamma_n}[g(u)],\]
where $\tilde \gamma_n$ is the density of a $N\left(\frac{2}{3}X^{(n)}, \frac{1}{3}I_n\right)$.
As the proportionality constants are the same in the numerator and in the denominator of \eqref{eq:ex1_limit},
\begin{equation}
\rho_{\mu,r}(X) = \lim_{n\rightarrow \infty} \frac{E_{\tilde \gamma_n}[-u_1 1_{\|u-X^{(n)}\| \leq r}]}{E_{\tilde \gamma_n}[1_{\|u-X^{(n)}\| \leq r} ]} = \lim_{n\rightarrow\infty}E_{\tilde \gamma_n}[-u_1 \mid \|u-X^{(n)}\| \leq r]\label{eq:ex1_conditional}
\end{equation}
Defining $V^{(n)} = u - X^{(n)}$, $m= -\frac{1}{3}e_1$ and $\sigma^2 = \frac{1}{3}$ we have that $V^{(n)} \sim N(m, \sigma^2 I_n)$. Hence, the condition $\|u-X^{(n)}\| \leq r$ on the right most equation in \eqref{eq:ex1_conditional} becomes $\|V^{(n)}\|\leq r$. Using the fact that $u_1 = V^{(n)}_1 +1$, we conclude that
\[\rho_{\mu,r}(X) = -1 - \lim_{n\rightarrow \infty} E[V_1^{(n)} \mid \|V^{(n)}\| \leq r].\]

By the definition of conditional expectation,
\[
  E[V^{(n)}_1 \mid \|V^{(n)}\| \leq r]
  \;=\;
  \frac{E\!\left[V^{(n)}_1\,1_{\{\|V^{(n)}\|\leq r\}}\right]}{P(\|V^{(n)}\|\leq r)}.
\] Write $V^{(n)} = m + \sigma Z$ with $Z \sim N(0, I_n)$.  Then
\[
  \frac{\|V^{(n)}\|^2}{\sigma^2}
  \;=\;
  \left\|Z + \frac{m}{\sigma}\right\|^2,
\]
which has the non-central chi-squared distribution $\chi^2_n(\lambda)$ with
$n$ degrees of freedom and non-centrality parameter
$\lambda = \|m\|^2/\sigma^2$.  Writing $F_k(x;\lambda)$ for the cdf of
$\chi^2_k(\lambda)$, we obtain the following expression for the denominator
\[
  P(\|V^{(n)}\|\leq r) \;=\; F_n\!\left(\tfrac{r^2}{\sigma^2};\,\lambda\right).
\] For the numerator, we first notice the density of $V^{(n)}$ is given by
\[
  f_{V^{(n)}}(v)
  \;=\;
  \frac{1}{(2\pi\sigma^2)^{n/2}}
  \exp\!\left(-\frac{\|v-m\|^2}{2\sigma^2}\right).
\]
Differentiating with respect to $m_1$ gives
\[
  \frac{\partial}{\partial m_1} f_{V^{(n)}}(v)
  \;=\;
  \frac{v_1 - m_1}{\sigma^2}\,f_{V^{(n)}}(v),
\]
so that
\[
  v_1\,f_{V^{(n)}}(v)
  \;=\;
  \sigma^2\,\frac{\partial}{\partial m_1} f_{V^{(n)}}(v)
  \;+\;
  m_1\,f_{V^{(n)}}(v).
\]
Integrating over the ball $\{\|v\|\leq r\}$ and using dominated convergence to
interchange differentiation with integration, we obtain
\begin{equation}\label{eq:numerator}
  E\!\left[V_1\,1_{\{\|V^{(n)}\|\leq r\}}\right]
  \;=\;
  \sigma^2\,\frac{\partial}{\partial m_1} P(\|V^{(n)}\|\leq r)
  \;+\;
  m_1\,P(\|V^{(n)}\|\leq r).
\end{equation} Since $\lambda = \|m\|^2/\sigma^2$ and only $m_1$ is non-zero, we have
$\partial\lambda/\partial m_1 = 2m_1/\sigma^2$.  Applying the chain rule and
\eqref{eq:cdf-deriv} with $x = r^2/\sigma^2$:
\[
  \frac{\partial}{\partial m_1} P(\|V^{(n)}\|\leq r)
  \;=\;
  \frac{\partial F_n(r^2/\sigma^2;\,\lambda)}{\partial\lambda}
  \cdot\frac{2m_1}{\sigma^2}
  \;=\;
  \frac{m_1}{\sigma^2}
  \left[F_{n+2}\!\left(\tfrac{r^2}{\sigma^2};\lambda\right)
        - F_n\!\left(\tfrac{r^2}{\sigma^2};\lambda\right)\right].
\] Substituting into \eqref{eq:numerator} and writing $F_k$ as shorthand for
$F_k(r^2/\sigma^2;\,\lambda)$:
\[
  E\!\left[V^{(n)}_1\,1_{\{\|V^{(n)}\|\leq r\}}\right]
  \;=\;
  \sigma^2 \cdot \frac{m_1}{\sigma^2}\!\left[F_{n+2} - F_n\right]
  \;+\;
  m_1\,F_n
  \;=\;
  m_1\!\left[F_{n+2} - F_n + F_n\right]
  \;=\;
  m_1\,F_{n+2}.
\]
Dividing by the denominator $P(\|V^{(n)}\|\leq r) = F_n$ gives the result:
\[
  E[V^{(n)}_1 \mid \|V^{(n)}\|\leq r]
  \;=\;
  \frac{m_1\,F_{n+2}}{F_n}
  \;=\;
  m_1\,
  \frac{P\!\left(\chi^2_{n+2}(\lambda)\leq r^2/\sigma^2\right)}
       {P\!\left(\chi^2_{n  }(\lambda)\leq r^2/\sigma^2\right)}.
  \qedhere
\]


Hence, the risk measure can be written as
\[
  \rho_{\mu,r}(X)
  \;=\;
  -1 + \lim_{n\rightarrow \infty} \frac{1}{3}
  \frac{F_{n+2}(3r^2;\,1/3)}{\,F_{n}(3r^2;\,1/3)}.
\]
Taking $n \to \infty$ for fixed $r$, Proposition \ref{prop:cdf-ratio} gives $F_{n+2}(3r^2;\,1/3)/F_n(3r^2;\,1/3) \sim 3r^2/n \to 0$,
so
\[
  \rho_{\mu,r}(X)
  \;=\;
  -1 + \lim_{n\to\infty}\frac{1}{3}\frac{3r^2}{n}
  \;=\;
  -1
\]
for every finite $r > 0$, which is consistent with the theoretical results of Lemma \ref{lmm:r}. We now examine the behaviour of
$\rho_{\mu,r}(X) = -1$ as $r$ varies, having already taken $n \to \infty$.
 
\medskip
\noindent\textbf{Limit $r \to 0$.}  As $r \to 0$, the ball $\bar{B}(X,r)$
shrinks to the single point $X = e_1$.  Since $\rho_{\mu,r}(X) = -1$ for
all finite $r > 0$, the limit is $-1$, which equals $\rho(X) = -X_1 = -1$,
the value of the risk functional evaluated at the centre of the ball.  This
is consistent: as the ball shrinks, the conditional measure concentrates on
$X$ and the risk measure converges to $\rho(X)$. It also illustrates the small ball limit of Theorem \ref{thm:r}.
 
\medskip
\noindent\textbf{Limit $r \to \infty$.}  Since $\rho_{\mu,r}(X) = -1$ for
every finite $r > 0$, the limit as $r \to \infty$ is again $-1$, also in line with the large ball limit of Theorem \ref{thm:r}. This is, however, in sharp contrast with what happens at finite $n$: for fixed $n$, as $r \to
\infty$ the ball covers all of $\mathbb{R}^n$, the conditioning becomes
vacuous, and
\[
  \rho_{\mu,r}^{(n)}(X)
  \;\to\;
  E[-V_1]
  \;=\;
  -(m_1+1)
  \;=\;
  -\tfrac{2}{3}.
\]
The two limits therefore do not commute:
\[
  \lim_{r\to\infty}\lim_{n\to\infty}\,\rho_{\mu,r}^{(n)}(X)
  \;=\;
  -1
  \;\neq\;
  -\tfrac{2}{3}
  \;=\;
  \lim_{n\to\infty}\lim_{r\to\infty}\,\rho_{\mu,r}^{(n)}(X).
\]
This is a manifestation of the curse of dimensionality: in high dimensions,
any ball of fixed radius captures a vanishing fraction of the Gaussian
measure, so the conditioning concentrates entirely on the centre of the ball
regardless of how large $r$ is, giving $\rho_{\mu,r}(X) = -1$ for all
finite $r > 0$.



\end{Exm}

\begin{Exm}[Quadratic risk functional]
We now consider the same setup as in Section \ref{sec:gauss_linear} with $\varphi(t)=e^{-t^2}$,
$X=e_1=(1,0,0,\ldots)$, and $\lambda_1=\lambda_2=\cdots=1$, but replace the
linear functional by the quadratic one $\rho(Z)=f(Z)=-Z_1^2$, for which the Lipschitz hypothesis of Theorem~\ref{thm:r} does not hold globally. Following the same finite-dimensional approximation as in \eqref{eq:ex1_limit}, we need to compute
\[
\rho_{\mu,r}(X)
= \lim_{n\to\infty}
\frac{E_{\gamma_n}\!\left[-u_1^2\,e^{-\|u-X^{(n)}\|^2}\,
      1_{\|u-X^{(n)}\|\le r}\right]}
{E_{\gamma_n}\!\left[e^{-\|u-X^{(n)}\|^2}\,
      1_{\|u-X^{(n)}\|\le r}\right]},
\]
where $u\sim\gamma_n=\mathcal{N}(0,I_n)$.

Exactly as in the linear case, the exponential weight
combined with the $\mathcal{N}(0,I_n)$ density tilts the measure to
$\tilde{\gamma}_n=\mathcal{N}\!\left(\tfrac{2}{3}X^{(n)},\tfrac{1}{3}I_n\right)$,
with the proportionality constants canceling between numerator and denominator. Hence,
\[
\rho_{\mu,r}(X)
= \lim_{n\to\infty}
E_{\tilde{\gamma}_n}\!\left[-u_1^2 \;\middle|\; \|u-X^{(n)}\|\le r\right].
\]

Setting $V^{(n)}=u-X^{(n)}$, we have
$V^{(n)}\sim\mathcal{N}\!\left(-\tfrac{1}{3}e_1,\tfrac{1}{3}I_n\right)$
and $u_1=V_1^{(n)}+1$, so
\[
u_1^2 = \bigl(V_1^{(n)}+1\bigr)^2
      = \bigl(V_1^{(n)}\bigr)^2 + 2V_1^{(n)} + 1.
\]
Therefore,
\[
\rho_{\mu,r}(X)
= -1
  - 2\lim_{n\to\infty}E\!\left[V_1^{(n)}\mid\|V^{(n)}\|\le r\right]
  -  \lim_{n\to\infty}E\!\left[\bigl(V_1^{(n)}\bigr)^2\mid\|V^{(n)}\|\le r\right].
\]

The term $\lim_{n\to\infty}E[V_1^{(n)}\mid\|V^{(n)}\|\le r]$ is exactly the one
appearing in the linear example and it is equal to $0$ for every finite $r>0$. For the second term, note that $\sqrt{3}\,V_1^{(n)}\sim\mathcal{N}(-1/\sqrt{3},1)$, so
$3(V_1^{(n)})^2$ is the squared first component of a vector whose squared
norm $3\|V^{(n)}\|^2$ follows a non-central $\chi^2(n)$ distribution with
non-centrality parameter $\lambda=1/3$. By the curse-of-dimensionality
argument used by the end of Section \ref{sec:gauss_linear} $-$ in high dimensions, any ball of fixed
radius captures a vanishing fraction of the Gaussian measure, so the
conditioning becomes asymptotically vacuous $-$ the conditional second moment
converges to the unconditional one:
\[
\lim_{n\to\infty}
E\!\left[\bigl(V_1^{(n)}\bigr)^2\mid\|V^{(n)}\|\le r\right]
= E\!\left[\bigl(V_1^{(n)}\bigr)^2\right]
= \mathrm{Var}(V_1^{(n)}) + \bigl(E[V_1^{(n)}]\bigr)^2
= \frac{1}{3}+\frac{1}{9}
= \frac{4}{9}.
\] Combining the three terms, we obtain
\[
\rho_{\mu,r}(X) = -1 - 0 - \frac{4}{9} = -\frac{13}{9}
\]
for every finite $r>0$.

The two limits in $r$ and $n$ do not commute, mirroring exactly the
phenomenon observed in the linear case. In the order $\lim_{r\to 0}
\lim_{n\to\infty}$, the result is $-13/9$ for all finite $r>0$.
In the reversed order, for fixed $n$ and $r\to 0$ the ball shrinks to the
single point $X=e_1$ and
$\rho_{\mu,r}^{(n)}(X)\to\rho(X)=-X_1^2=-1$,
which is the value expected from \Cref{thm:r}.
This non-commutativity is again a manifestation of the curse of dimensionality:
in the infinite-dimensional limit, any ball of fixed radius concentrates its
$\gamma$-mass entirely at the center, so the conditioning is insensitive to
$r$ and never recovers the base risk measure value $\rho(X)=-1$.
\end{Exm}

\begin{Exm}[A Bayesian example: Gaussian families under the Wasserstein metric]
\label{sec:bayesian}
The preceding examples studied $\rho_{\mu,r}$ in the Hilbert-space setting of
Example \ref{ex:gaus}, where the ambient metric is induced by the norm and the
theoretical machinery of Theorems \ref{thm:r} applies in full.
We now shift perspective and work \emph{directly at the level of distributions},
placing ourselves in the Law-Invariance framework of Section \ref{sec:agg}.
The goal is twofold: to provide a tractable, simulation-friendly
illustration of the aggregation bounds in Corollary \ref{crl:avg}, and to show
that the averaging approach is well-calibrated relative to both the
baseline risk measure and the worst-case benchmark as the uncertainty
radius $r$ varies.
 
We restrict attention to the parametric family of one-dimensional
Gaussian distributions,
\[
  \mathcal{G} \;=\; \bigl\{Z \sim \mathcal{N}(\mu_Z, \sigma_Z^2)
    : \mu_Z \in \mathbb{R},\; \sigma_Z > 0\bigr\},
\]
so that every element of the uncertainty neighborhood is identified
by its mean--standard-deviation pair $(\mu_Z, \sigma_Z)$.
We endow $\mathcal{G}$ with the \emph{$2$-Wasserstein distance}, which
takes the closed form
\begin{equation}
  d_{\mathcal{W}_2}(X, Z)
  \;=\;
  \sqrt{(\mu_X - \mu_Z)^2 + (\sigma_X - \sigma_Z)^2},
  \label{eq:W2_Gaussian}
\end{equation}
where we identify each distribution with its parameters.
The closed uncertainty ball of radius $r$ around $X \sim \mathcal{N}(\mu_X, \sigma_X^2)$
is therefore the Euclidean disk
\[
  \bar{B}(X, r)
  \;=\;
  \bigl\{(\mu_Z, \sigma_Z) \in \mathbb{R} \times (0,\infty)
    : (\mu_X - \mu_Z)^2 + (\sigma_X - \sigma_Z)^2 \leq r^2\bigr\}.
\] This setting corresponds to the Wasserstein framework of Example \ref{ex:wass}
  rather than to the Banach-lattice setting of Theorem \ref{Thm:main}, so the
  lower bound $\rho \le \rho_{\mu,r}$ is not guaranteed by our theory.
  Whether it holds for specific risk measures and prior choices is
  therefore an empirical question, and one we examine numerically below.

As the base measure $\gamma_X$ we use a \emph{Normal-Gamma} distribution
for the pair $(\mu_Z, \tau_Z)$, where $\tau_Z: = \sigma_Z^{-2}$ denotes
the precision.  This is the standard conjugate prior for a Gaussian
likelihood with unknown mean and variance, and it ensures that
$\gamma_X$ is concentrated on $\mathcal{G}$.
Specifically, we write
$(\mu_Z, \tau_Z) \sim \mathrm{NormalGamma}(\mu_X, k, \alpha_{\mathrm{NG}}, \beta_{\mathrm{NG}})$,
meaning
\[
  \mu_Z \mid \tau_Z \;\sim\; \mathcal{N}\!\bigl(\mu_X,\,(k\tau_Z)^{-1}\bigr),
  \qquad
  \tau_Z \;\sim\; \mathrm{Gamma}(\alpha_{\mathrm{NG}}, \beta_{\mathrm{NG}}),
\]
with density
\[
  p(\mu_Z, \tau_Z)
  \;=\;
  \frac{\beta_{\mathrm{NG}}^{\alpha_{\mathrm{NG}}}\sqrt{k}}{\Gamma(\alpha_{\mathrm{NG}})\sqrt{2\pi}}\,
  \tau_Z^{\alpha_{\mathrm{NG}} - \frac{1}{2}}\,
  e^{-\beta_{\mathrm{NG}}\tau_Z}\,
  \exp\!\Bigl(-\tfrac{k\tau_Z}{2}(\mu_Z - \mu_X)^2\Bigr).
\]
We center the prior at the baseline distribution $X$ by setting $\beta_{\mathrm{NG}} = \alpha_{\mathrm{NG}}/\tau_X$, so that
$\mathbb{E}[\mu_Z] = \mu_X$ and $\mathbb{E}[\tau_Z] = \tau_X:=\sigma^{-2}_X$. Under this parametrization, the marginal variance of $\mu_Z$
is $\text{Var}(\mu_Z) = \beta_{\mathrm{NG}} / (k(\alpha_{\mathrm{NG}}-1))$ and  $\text{Var}(\tau_Z) =\tau_X^2/\alpha_{\mathrm{NG}}$,
so $\alpha_{\mathrm{NG}}$ and $k$ jointly control how tightly
the prior concentrates around $(\mu_X, \tau_X)$. The hyperparameters $\alpha_{\mathrm{NG}}$ and $k$ control the \emph{concentration}
of the prior: large $\alpha_{\mathrm{NG}}$ (or large $k$) concentrates mass near
$(\mu_X, \tau_X)$, mimicking the effect of a small radius $r$.
This parallel between Bayesian concentration and uncertainty-ball
shrinkage is a useful interpretive bridge: the radius $r$ and the
prior precision $(\alpha_{\mathrm{NG}}, k)$ are complementary ways of encoding
the practitioner's confidence in the baseline model.
 
We take the Expected Shortfall at level $a \in (0,1)$ as the 
baseline risk measure. For $Z \sim \mathcal{N}(\mu_Z, \sigma_Z^2)$ 
it reads
\begin{equation}\label{eq:gaussian_ES}
  \rho(Z) \;=\; \mathrm{ES}_a(Z)
  \;=\; -\mu_Z + \sigma_Z\,\frac{\phi_{\mathrm{n}}(\Phi^{-1}(a))}{1 - a},
\end{equation}
where $\phi_{\mathrm{n}}$ and $\Phi$ denote the standard Gaussian 
pdf and cdf, respectively; we use the subscript $\mathrm{n}$ to 
distinguish $\phi_{\mathrm{n}}$ from the spectral weight function 
$\phi$ of \Cref{sec:agg}. Since $\mathrm{ES}_a$ is a spectral risk measure, it is 
in particular concave on $\mathcal{M}(\mathcal{X})$. Theorem \ref{thm:avg}
therefore yields both the equality 
$\rho_{\mu,r}(X) = \rho(X_{\bar{\mu},r})$ and the inequality 
$\rho_{\mu,r}(X) \le \rho(X_{\mu,r})$. Whether the lower 
bound from Corollary \ref{crl:avg}, which relies on the Banach-lattice 
setting of Theorem \ref{Thm:main}, extends to the present Wasserstein 
setting is not guaranteed by our theory, and we verify it 
numerically:
\begin{equation}
  \rho(X) \;\stackrel{?}{\le}\; \rho_{\mu,r}(X) 
  \;=\; \rho(X_{\bar{\mu},r})
  \;\le\; \rho(X_{\mu,r}).
  \label{eq:dominance_chain}
\end{equation} The weighting kernel is chosen as $\varphi(t) = e^{-\lambda t^2}$,
giving more importance to distributions close to $X$ in the
Wasserstein sense.  The normalizing constant is
\begin{align*}
  K(X, r)
  &= \int_{\bar{B}(X,r)}
       \exp\!\Bigl(-\lambda\bigl[(\mu_X-\mu_Z)^2
         + (\sigma_X - \sigma_Z)^2\bigr]\Bigr)\,
       p(\mu_Z, \tau_Z)\,d\mu_Z\,d\tau_Z,
\end{align*}
and the robust risk measure is
\begin{align}
  \rho_{\mu,r}(X)
  &= \frac{1}{K(X,r)}
     \int_{\bar{B}(X,r)}
       \Bigl(-\mu_Z + \sigma_Z\,
         \tfrac{\phi_n(\Phi^{-1}(a))}{1-a}\Bigr)\,
       \exp\!\Bigl(-\lambda\bigl[(\mu_X-\mu_Z)^2
         +(\sigma_X-\sigma_Z)^2\bigr]\Bigr)\notag\\
  &\quad\times p(\mu_Z, \tau_Z)\,d\mu_Z\,d\tau_Z.
  \label{eq:rho_bayesian}
\end{align} 

Both integrals are evaluated by Monte Carlo:
we draw $N$ samples $(\mu_Z^{(i)}, \tau_Z^{(i)})$ from the
Normal-Gamma prior, retain those falling inside $\bar{B}(X, r)$,
apply the weights $e^{-\lambda d_{\mathcal{W}_2}(X,Z^{(i)})^2}$,
and compute the weighted average of $\rho(Z^{(i)})$. The baseline law has parameters $\mu_X = 0$ and $\sigma_X = 1$
(so $\tau_X = 1$), and the ES level is $a = 0.95$.
The kernel decay is $\lambda = 2$, and all Monte Carlo estimates
use $N = 10^6$ draws over the radius grid $r \in [0, 2]$.
The Normal-Gamma prior is specified by shape
$\alpha_{\mathrm{NG}} = 25$, rate
$\beta_{\mathrm{NG}} = \alpha_{\mathrm{NG}}/\tau_X = 25$,
and precision scaling $k = 4$; these correspond to a prior
standard deviation of $20\%$ on the precision $\tau_Z$ and
$0.5\,\sigma_X$ on the mean $\mu_Z$.

To numerically verify \eqref{eq:dominance_chain} we also compute
the two aggregated benchmarks of \Cref{def:diff_aggregations}. For the \emph{quantile-aggregated} quantity $X_{\bar{\mu},r}$,
we use the fact that, within the Gaussian family, the quantile
function of $Z\sim\mathcal{N}(\mu_Z,\sigma_Z^2)$ is
$F_Z^{-1}(u)=\mu_Z+\sigma_Z\,\Phi^{-1}(u)$.
The averaged quantile function \eqref{eq:fmuq} therefore reads
\begin{equation}\label{eq:quantile_agg_gaussian}
  F_{\bar{\mu},r}^{-1}(u)
  \;=\;
  \int_{\bar{B}(X,r)} F_Z^{-1}(u)\,d\mu_{X,r}(Z)
  \;=\;
  \bar{\mu}_{\mu,r} + \bar{\sigma}_{\mu,r}\,\Phi^{-1}(u),
\end{equation}
where $\bar{\mu}_{\mu,r}$ and $\bar{\sigma}_{\mu,r}$ denote the
$\mu_{X,r}$-weighted averages of $\mu_Z$ and $\sigma_Z$ over
$\bar{B}(X,r)$, both computed by Monte Carlo alongside
$\rho_{\mu,r}(X)$.
Since $X_{\bar{\mu},r}\sim\mathcal{N}(\bar{\mu}_{\mu,r},\bar{\sigma}_{\mu,r}^2)$,
its ES is available in closed form:
\begin{equation}\label{eq:ES_quantile_agg}
  \rho(X_{\bar{\mu},r})
  \;=\; -\bar{\mu}_{\mu,r}
    + \bar{\sigma}_{\mu,r}\,
      \frac{\phi_n(\Phi^{-1}(a))}{1-a}.
\end{equation}
By \Cref{thm:avg}, $\rho(X_{\bar{\mu},r})=\rho_{\mu,r}(X)$, so
\eqref{eq:ES_quantile_agg} provides a closed-form cross-check
for the Monte Carlo estimate of $\rho_{\mu,r}(X)$.

For the \emph{distribution-aggregated} quantity $X_{\mu,r}$,
we form the mixture cdf \eqref{eq:fmu}, which here becomes a
Gaussian location-scale mixture:
\[
  F_{\mu,r}(x)
  \;=\; \int_{\bar{B}(X,r)}
        \Phi\!\left(\frac{x - \mu_Z}{\sigma_Z}\right)
        d\mu_{X,r}(Z).
\]
This mixture is not Gaussian in general, so
$\rho(X_{\mu,r})=\mathrm{ES}_a(X_{\mu,r})$ is computed numerically
from the empirical quantile of the simulated sample
$\{F_{\mu,r}^{-1}(U^{(i)})\}_{i=1}^N$. For reference we also compute the worst-case risk measure
\[
  \rho^{WC}(X)
  \;=\; \sup_{Z \in \bar{B}(X,r)} \rho(Z)
  \;=\; \sup_{\substack{(\mu_Z,\sigma_Z):\,\sigma_Z>0,\\
        d_{\mathcal{W}_2}(X,Z)\le r}}
    \Bigl(-\mu_Z + \sigma_Z\,c_a\Bigr),
\]
where we write $c_a := \phi_n(\Phi^{-1}(a))/(1-a)$ for
brevity. Since $\rho(Z)$ is affine in $(\mu_Z,\sigma_Z)$, the
supremum can be computed in closed form. Setting
$\mathbf{v}=(\mu_Z,\sigma_Z)^\top$,
$\mathbf{v}_0=(\mu_X,\sigma_X)^\top$, and
$\mathbf{w}=(-1,c_a)^\top$, the problem reduces to
\[
  \sup_{\|\mathbf{v}-\mathbf{v}_0\|\le r} \mathbf{w}^\top\mathbf{v}.
\]
Decomposing $\mathbf{w}^\top\mathbf{v}
= \mathbf{w}^\top\mathbf{v}_0
+ \mathbf{w}^\top(\mathbf{v}-\mathbf{v}_0)$
and applying Cauchy--Schwarz gives
\[
  \mathbf{w}^\top\mathbf{v}
  \;\le\;
  \mathbf{w}^\top\mathbf{v}_0
  + \|\mathbf{w}\|\,\|\mathbf{v}-\mathbf{v}_0\|
  \;\le\;
  \mathbf{w}^\top\mathbf{v}_0 + r\|\mathbf{w}\|,
\]
with equality when $\mathbf{v}-\mathbf{v}_0
= r\,\mathbf{w}/\|\mathbf{w}\|$, i.e.\ the optimizer moves
from the center $\mathbf{v}_0$ in the direction of the gradient
$\mathbf{w}$. Since
$\mathbf{w}^\top\mathbf{v}_0 = -\mu_X + c_a\sigma_X = \rho(X)$
and $\|\mathbf{w}\|=\sqrt{1+c_a^2}$, this yields
\begin{equation}\label{eq:WC_closed_form}
  \rho^{WC}(X)
  \;=\; \rho(X) + r\sqrt{1+c_a^2}.
\end{equation}
The positivity constraint $\sigma_Z>0$ is automatically satisfied
at the optimizer: the optimal value $\sigma_Z^*
= \sigma_X + rc_a/\|\mathbf{w}\|$ is strictly positive for all
$r\ge 0$ since $c_a>0$. Hence \eqref{eq:WC_closed_form} is valid
for all $r\ge 0$, and $\rho^{WC}(X)$ grows linearly in $r$ at
rate $\sqrt{1+c_a^2}$, providing a rigorous upper envelope against
which to compare the less conservative averaging measure.

\begin{figure}
    \centering    \includegraphics[width=0.9\linewidth]{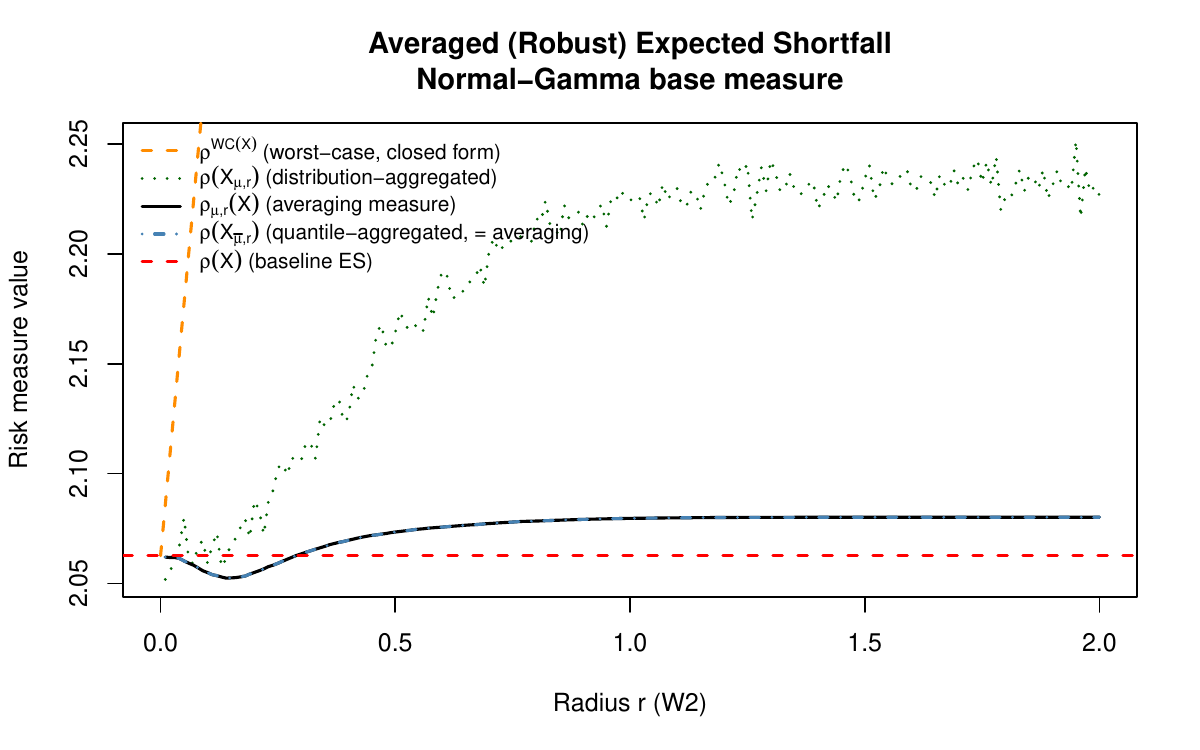}
\caption{Dominance chain for the averaging measure under the
  Normal-Gamma base measure with $\mu_X=0$, $\sigma_X=1$,
  $a=0.95$, $\lambda=2$, $\alpha_{\mathrm{NG}}=25$, $k=4$,
  $N=10^6$. Five quantities are plotted as functions of the
  Wasserstein radius $r$: baseline ES $\rho(X)$ (red dashed),
  averaging measure $\rho_{\mu,r}(X)$ (solid black),
  quantile-aggregated ES $\rho(X_{\bar{\mu},r})$ (dot-dashed
  blue, numerically coinciding with $\rho_{\mu,r}(X)$ in
  agreement with \Cref{thm:avg}), distribution-aggregated ES
  $\rho(X_{\mu,r})$ (dotted green), and worst-case ES
  $\rho^{WC}(X)$ (dashed orange, closed form
  \eqref{eq:WC_closed_form}).}
    \label{fig:ng_ES_dominance}
\end{figure}

\Cref{fig:ng_ES_dominance} displays five quantities as functions of
$r\in[0,2]$: the baseline $\rho(X)$ (red dashed, constant),
the averaging measure $\rho_{\mu,r}(X)$ (solid black),
the quantile-aggregated $\rho(X_{\bar{\mu},r})$ (dot-dashed blue,
equal to $\rho_{\mu,r}(X)$ by \Cref{thm:avg}),
the distribution-aggregated $\rho(X_{\mu,r})$ (dotted green),
and the worst-case $\rho^{WC}(X)$ (dashed orange, from
\eqref{eq:WC_closed_form}). Several features in this figure deserve comment. First, as $r\downarrow 0$ all five quantities converge to $\rho(X)$. In particular, $\rho_{\mu,r} \rightarrow \rho(X)$ when $r\rightarrow 0$, consistent with the small radius result from \Cref{thm:r}. When $r \rightarrow \infty$ (large radius limit), the proposed robust risk measure converges to a constant. Second, the solid black and dot-dashed blue curves are numerically 
indistinguishable up to Monte Carlo error of order $O(N^{-1/2})$,
providing empirical support for the identity
$\rho_{\mu,r}(X)=\rho(X_{\bar{\mu},r})$ from \Cref{thm:avg}.
It also serves as an implementation cross-check: since the two
quantities are computed by different code paths --- a direct
weighted average of $\rho(Z^{(i)})$ versus the closed-form
\eqref{eq:ES_quantile_agg} evaluated at $(\bar{\mu}_{\mu,r},
\bar{\sigma}_{\mu,r})$ --- their agreement up to sampling
noise confirms the correctness of both estimators. Third, we comment on the empirical dominance chain. Apart from very small values of $r$ (where Monte Carlo estimates are extremely noisy), the ordering
    \[
      \rho(X) \;\le\; \rho_{\mu,r}(X) \;\le\; \rho(X_{\mu,r})
      \;\le\; \rho^{WC}(X)
    \]
    holds numerically.  The middle two inequalities are guaranteed
    by \Cref{thm:avg} and \Cref{crl:avg}; the leftmost one
    (marked $\stackrel{?}{\le}$ in \eqref{eq:dominance_chain})
    is not covered by our theory in the Wasserstein setting but is
    confirmed here empirically, suggesting that $\rho\le\rho_{\mu,r}$
    may hold more generally than \Cref{Thm:main} currently guarantees. Finally, the averaging measure $\rho_{\mu,r}(X)$ grows smoothly and sub-linearly in $r$, while $\rho^{WC}(X)$ grows at the constant rate $\sqrt{1+(\phi_n(\Phi^{-1}(a))/(1-a))^2}$ per unit of $r$.  This confirms the main motivation for the averaging approach: it provides robustness without the excessive conservatism of the worst-case.

Figure~\ref{fig:bayesian_sensitivity} isolates the effect of each tuning parameter in turn, keeping the dominance-chain quantities fixed at the main-experiment values for reference.

\begin{figure}[!ht]
  \centering
  \includegraphics[width=\textwidth]{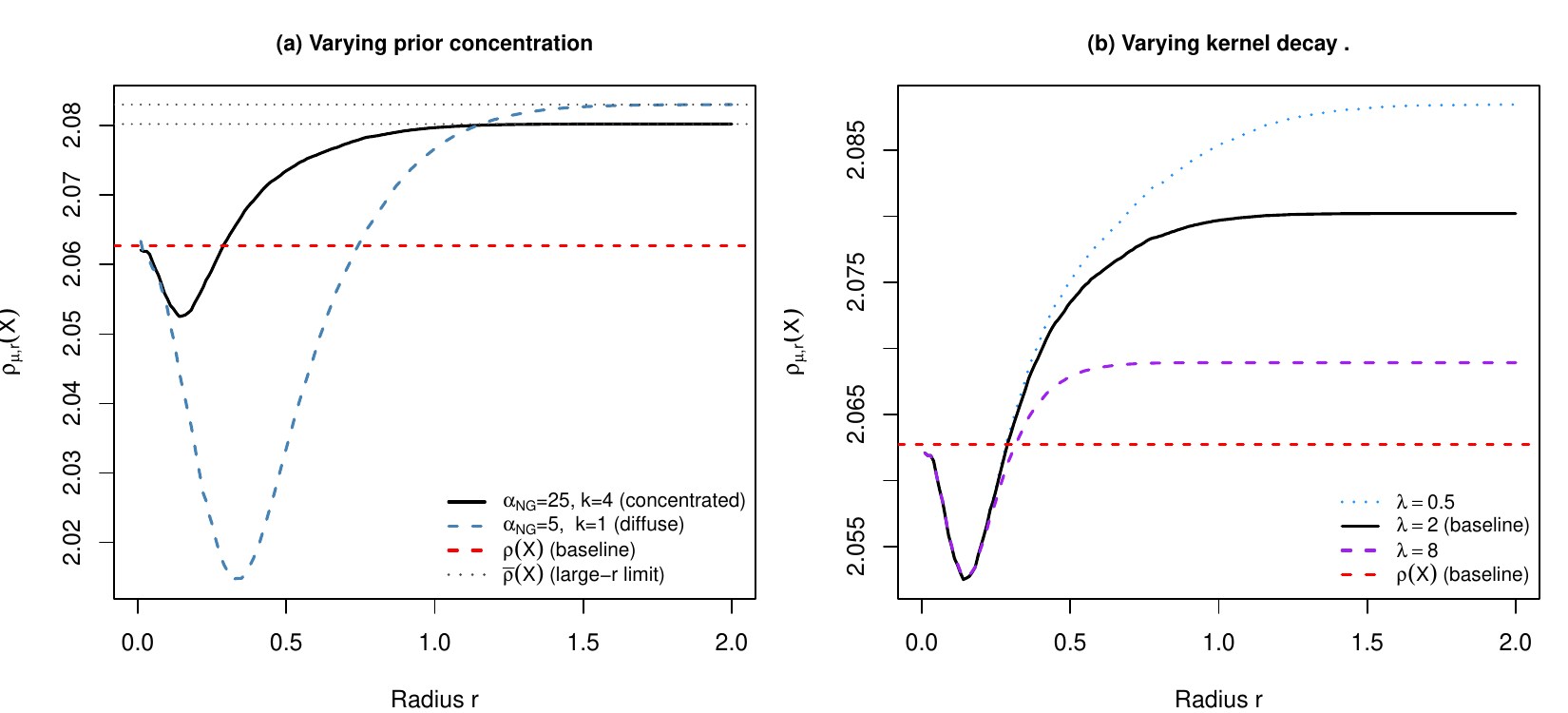}
  \caption{Sensitivity of the averaging measure $\rho_{\mu,r}(X)$
    to tuning parameters ($\mu_X=0$, $\sigma_X=1$, $a=0.95$,
    $N=10^6$).
    \emph{Panel~(a)}: varying prior concentration with $\lambda=2$
    fixed; concentrated prior ($\alpha_{\mathrm{NG}}=25$, $k=4$,
    solid black) versus diffuse prior ($\alpha_{\mathrm{NG}}=5$,
    $k=1$, dashed blue).  Horizontal dotted lines mark the
    respective large-$r$ limits $\bar\rho(X)$ of \Cref{thm:r}.
    \emph{Panel~(b)}: varying kernel decay $\lambda\in\{0.5,2,8\}$
    with the concentrated prior fixed.
    Baseline ES $\rho(X)$ shown as red dashed line in both panels.}
  \label{fig:bayesian_sensitivity}
\end{figure}
 Panel~(a) fixes $\lambda=2$ and contrasts the concentrated prior ($\alpha_{\mathrm{NG}}=25$, $k=4$) with a diffuse one ($\alpha_{\mathrm{NG}}=5$, $k=1$).  The diffuse prior assigns appreciable weight to distributions far from $X$ even for small $r$, so $\rho_{\mu,r}(X)$ rises more steeply near $r=0$ and reaches its large-$r$ limit $\bar\rho(X)$ of \Cref{thm:r} sooner.  The concentrated prior keeps most mass near $(\mu_X,\sigma_X)$ until $r$ is large enough to include the tails of the prior, producing a slower and more gradual increase. The two limits $\bar\rho(X)$ differ between the priors (shown as horizontal dotted lines), since $\bar\rho$ integrates $\rho$ against the tilted measure $\bar\gamma_0$ of \Cref{thm:r}, which depends on $\gamma_X$. Panel~(b) fixes the concentrated prior and varies $\lambda\in\{0.5,2,8\}$. Larger $\lambda$ concentrates $\mu_{X,r}$ near $X$, making $\rho_{\mu,r}(X)$ closer to $\rho(X)$ for any given $r$; in the limit $\lambda\to\infty$ the kernel degenerates to a point mass at $X$ and $\rho_{\mu,r}(X)\to\rho(X)$ for every $r$, recovering the base risk measure. Conversely, $\lambda\to 0$ (uniform kernel) yields the most dispersed average.  The three curves in Panel~(b) span a non-trivial range, confirming that $\lambda$ is not redundant with $r$: unlike $r$, which shifts the \emph{reach} of the uncertainty neighborhood, $\lambda$ controls the \emph{shape} of the weighting within it. A practitioner therefore has two orthogonal robustness dials, and $(\alpha_{\mathrm{NG}}, k)$ can be fixed on substantive grounds while $r$ and $\lambda$ are calibrated to the desired level and profile of robustness.
\end{Exm}

\section*{Funding}
Righi acknowledges financial support from CNPq (Brazilian Research Council) project numbers 401720/2023-3 and 302869/2024-7.

\bibliography{Theory,reference}
\bibliographystyle{apalike}

\end{document}